\documentclass[11pt]{article}

\usepackage{odonnell}
\usepackage{graphicx}
\usepackage{subfigure}
\usepackage{epstopdf}

\newcommand{\x}{\bx}

\newcommand{\bigp}{\mathcal{P}}
\newcommand{\bigx}{\mathcal{X}}
\newcommand{\bigy}{\mathcal{Y}}

\newcommand{\ent}{\mathbf{H}}
\newcommand{\bfe}{\mathbf{E} }

\newtheorem{conjecture}[theorem]{Conjecture}

\begin{document}

\title{Decision Trees, Protocols, and\\ the Fourier Entropy-Influence Conjecture}

\author{Andrew Wan\thanks{On leave from IIIS, Tsinghua University.  This
work was completed at Harvard University and supported by NSF grant CCF-964401 and NSFC grant 61250110218.}\\ Simons Institute, U.C.\ Berkeley \\ \tt{atw12@seas.harvard.edu} 
\and 
John Wright\thanks{Supported by NSF grants CCF-0747250 and CCF-1116594 and a grant from the MSR--CMU Center for Computational Thinking. Some of this research done while visiting the Toyota Technological Institute at Chicago.}\\ Carnegie Mellon University\\ \tt{jswright@cs.cmu.edu}
\and Chenggang Wu\thanks{This work was supported in part by the National Basic Research Program of China Grant 2011CBA00300, 2011CBA00301, the National Natural Science Foundation of China Grant 61033001, 61061130540. Research done while visiting Carnegie Mellon University.}\\ IIIS, Tsinghua University\\ \tt{wcg06@mails.tsinghua.edu.cn}}

\maketitle

\begin{abstract}
Given $f:\{-1, 1\}^n \rightarrow \{-1, 1\}$, define the \emph{spectral distribution} of $f$ to be the distribution on subsets of $[n]$ in which the set $S$ is sampled with probability $\widehat{f}(S)^2$.  Then the \emph{Fourier Entropy-Influence (FEI) conjecture} of Friedgut and Kalai~\cite{FK96} states that there is some absolute constant $C$ such that $\ent[\widehat{f}^2] \leq C\cdot\Inf[f]$.  Here, $\ent[\widehat{f}^2]$ denotes the Shannon entropy of $f$'s spectral distribution, and $\Inf[f]$ is the total influence of $f$.  This conjecture is one of the major open problems in the analysis of Boolean functions, and settling it would have several interesting consequences.

Previous results on the FEI conjecture have been largely through direct calculation.  In this paper we study a natural interpretation of the conjecture, which states that there exists a communication protocol which, given subset $S$ of $[n]$ distributed as $\widehat{f}^2$, can communicate the value of $S$ using at most $C\cdot\Inf[f]$ bits in expectation.
%\anote{I'm debating whether we should keep this as $C\cdot \Inf[f]$, or replace with: 'can communicate, on average, the value of S using at most a constant times the size of S', which sounds to me maybe more striking.}\jnote{I think we should either do $C\cdot \Inf[f]$ or $C\cdot \E[|S|]$. I'm leaning towards $C\cdot \Inf[f]$ because I think the part of the interpretation we should be emphasizing is the protocol part, and not necessarily what we're upper bounding the length with.}.
Using this interpretation, we are able show the following results:
\begin{itemize}
\item First, if $f$ is computable by a read-$k$ decision tree, then $\ent[\widehat{f}^2] \leq 9k\cdot \Inf[f]$.
\item Next, if $f$ has $\Inf[f] \geq 1$ and is computable by a decision tree with expected depth $d$, then $\ent[\widehat{f}^2] \leq 12d\cdot \Inf[f]$.
\item Finally, we give a new proof of the main theorem of O'Donnell and Tan~\cite{OT13}, i.e. that their FEI$^+$ conjecture composes.
\end{itemize}
In addition, we show that natural improvements to our decision tree results would be sufficient to prove the FEI conjecture in its entirety.  
We believe that our methods give more illuminating proofs than previous results about the FEI conjecture.
\end{abstract}

\newpage

\section{Introduction}

Given a Boolean function $f : \{-1, 1\}^n\rightarrow\{-1, 1\}$, define the \emph{spectral distribution} of $f$ to be the distribution on subsets of $[n]$ in which the set $S$ is sampled with probability $\widehat{f}(S)^2$.  Overloading notation, we will denote this distribution by $\widehat{f}^2$.  Write $\bigx\sim\widehat{f}^2$ for the random variable which is distributed according to $\widehat{f}^2$.  The \emph{Fourier Entropy-Influence (FEI) Conjecture} of Friedgut and Kalai~\cite{FK96} states that there is some absolute constant $C$ such that $\ent[\bigx] \leq C\cdot\Inf[f]$, where $\Inf[f]$ is the total influence of $f$, and $\ent[\bigx]$ is the spectral entropy of $f$ (equivalently, the Shannon entropy of $\bigx$), which equals
\begin{equation*}
\ent[\bigx] = \sum_{S \subseteq [n]} \widehat{f}(S)^2 \log \left(\frac{1}{\widehat{f}(S)^2}\right).
\end{equation*}
\noindent The FEI Conjecture has been shown to have several interesting consequences, including a learning algorithm for DNFs in the agnostic learning model~\cite{Man94,GKK08}, and resolving it is a central question in the analysis of Boolean functions.  See~\cite{OWZ11} for a comprehensive introduction to the subject.

Verifying the conjecture for individual functions---such as Majority, AND/OR, and Tribes---can be done via straightforward calculation.  Verifying it for larger classes of functions requires more subtle argumentation.  To date, it has been shown to hold for random DNFs~\cite{KLW10}, symmetric functions and read-once decision trees~\cite{OWZ11}, and read-once formulas~\cite{OT13, CKLS13}.  Unfortunately, this conjecture lends itself to proofs which are at times opaque and conceptually unilluminating.  Perhaps one of the reasons is that whereas the total influence $\Inf[f]$ is a central quantity in the analysis of Boolean functions, the spectral entropy $\ent[\bigx]$ is rarely encountered and poorly understood.

In this paper we consider the natural interpretation of the FEI conjecture as stating the existence of a coding scheme for the random variable $\bigx$ with a certain performance.  Roughly speaking, the coding scheme must use, on average, some fixed constant times the \textit{size} of $\bigx$ (see Section \ref{sec:shannon} for a precise description).    Using this interpretation, we give three results concerning the FEI conjecture; we believe that our proofs of these results are both straightforward and conceptually interesting.
%our proofs of these results are both straightforward and conceptually interesting.

For our first result, we verify the conjecture for read-$k$ decision trees, where $k$ is a constant.  This is the class of decision trees in which each variable is queried at no more than $k$ distinct locations in the entire tree.  Previous results---those for read-once decision trees~\cite{OWZ11} and read-once formulas~\cite{OT13, CKLS13}---failed to generalize even to the read-twice case, as allowing a decision tree to be read-twice introduces correlations between different parts of the tree, and this is difficult to analyze.  In this paper, we surmount this barrier, proving:
\begin{theorem}\label{thm:readk}
Suppose $f:\{-1, 1\}^n \rightarrow \{-1, 1\}$ can be computed by a read-$k$ decision tree, and let $\bigx \sim \widehat{f}^2$.  Then $\ent[\bigx] \leq 9k \cdot \Inf[f]$.
\end{theorem}
A natural question is whether this can be improved to show the FEI conjecture for read-$k(n)$ decision trees, where $k(n) = \omega(1)$ is a slowly growing function of $n$.  However, a simple padding argument shows that this would be sufficient to prove the full FEI conjecture: given $f:\{-1, 1\}^n \rightarrow \{-1, 1\}$, one could add enough dummy variables to $f$ so that $k(\cdot)$ is at least $2^n$, and since any $n$-variable function is trivially computable by a read-$2^n$ decision tree, $f$ would satisfy the FEI conjecture.

%For our next result, we verify the conjecture for decision trees with expected depth $d$, where $d$ is a constant.  We are able to do so using much of the same proof as for Theorem~\ref{thm:readk}.
Using much of the same proof as for Theorem \ref{thm:readk}, we then verify the conjecture for decision trees with expected depth $d$, where $d$ is a constant.  The FEI conjecture trivially holds for depth-$d$ decision trees, which have a bounded number of variables, and so what makes this interesting is that we only require a bound on the \emph{expected} depth of the tree.  Our result is:

\begin{theorem}\label{thm:depthd}
Suppose $f:\{-1, 1\}^n\rightarrow\{-1, 1\}$ is computable by a decision tree whose expected depth is $d$.  Further, suppose $\Inf[f] \geq 1$.  Then $\ent[\widehat{f}^2] \leq 12d \cdot \Inf[f]$.
\end{theorem}
As before, if we could show the FEI conjecture for decision trees with expected depth $d(n)$, where $d(n) = \omega(1)$ is a slowly growing function of $n$, we would be able to show the full FEI conjecture.  In addition, the requirement in Theorem~\ref{thm:depthd} that $\Inf[f]$ be reasonably large is necessary, as we show in Appendix~\ref{app:biginf}.  We note that this result (with a better constant) also follows from the bound $\ent[\widehat{f}^2] \leq 2d$, which was proven independently by~\cite{CKLS13}.

For our final result, we give a new proof of the main theorem from~\cite{OT13}, which is a composition theorem for the FEI conjecture.  Their main application is to verify the FEI conjecture for read-once formulas.  For example, consider trying to prove the FEI conjecture for a read-once DNF formula (an OR of ANDs).  It is easy to verify that both the AND and OR functions (of any input size) each individually satisfy the FEI conjecture, but it is not so obvious how to prove that their composition satisfies it.

More broadly, let $f$ and $g_1,\ldots,g_k$ be Boolean functions, and consider the composition $h = f(g_1, \ldots, g_k)$, where each $g_i$ is over its own set of variables.  Their paper considers the following question: supposing that $f$ and the $g_i$'s satisfy the FEI conjecture with constant $C$, what can one conclude about $h$?    Perhaps their main contribution is in noting that from $f$'s perspective, it is not receiving perfectly unbiased bits as inputs, but $\E[g_i]$-biased bits.   Thus it is natural that it shouldn't matter whether $f$ satisfies the FEI conjecture, but rather whether it satisfies some $\E[g_i]$-biased version of the FEI conjecture.  They formulate this biased version of the FEI conjecture, which they call the FEI$^+$ conjecture (which we will formally state later), and prove the following composition theorem:
\begin{theorem}[Informal]\label{thm:informal}
Suppose $f$ and $g_1,\ldots,g_k$ satisfy the FEI$^+$ conjecture with constant $C$.  Then $h = f(g_1, \ldots, g_k)$ also satisfies the FEI$^+$ conjecture with constant $C$.
\end{theorem}
They proved this by expanding the expressions $\ent[\bigx]$ and $\Inf[h]$ in terms of the Fourier coefficients of $f$ and $g_1,\ldots,g_k$, and comparing the results.  Using our coding theoretic interpretation of the FEI conjecture, we give a new proof of this theorem which shows that codes compose in a very clean way.

We now describe our interpretation of the FEI conjecture and discuss our main results in more detail.
\subsection{The FEI Conjecture as a Coding Bound}\label{sec:shannon}

Let $\bigx \sim \widehat{f}^2$.  We view the Fourier Entropy-Influence Conjecture as stating the existence of highly efficient coding schemes for communicating the value of $\bigx$.  To explain this, we begin with some standard information theory background.  Given a domain $\mathcal{D}$ and an output alphabet $\Sigma$, a \emph{code} on $\mathcal{D}$ is a function $c:\mathcal{D}\rightarrow \Sigma^*$.  We say that $c$ is \emph{prefix-free} if $c(x)$ is never a prefix of $c(y)$ for distinct $x, y \in \mathcal{D}$.  If $\x$ is a random variable which takes values in $\mathcal{D}$, then the average number of characters output by $c$, called the \emph{length} of $c$, is $\E[|c(\x)|]$, and we often care about finding a code $c$ which minimizes this quantity.  The source coding theorem of Shannon says that $\ent[\x]$ is roughly the best possible length achievable by a prefix-free code:

\begin{theorem}[Shannon's source coding theorem~\cite{Sha48}] \label{thm:shannon}
Let $\x$ be a random variable over a domain $\mathcal{D}$ and let $\Sigma$ be a finite alphabet.
\begin{enumerate}
\item If $c:\mathcal{D} \rightarrow \Sigma^*$ is a prefix-free code for $\x$, then
$\ent[\x]/\log_2|\Sigma| \leq \E[|c(\x)|]$.
\item Furthermore, there exists a prefix-free code $c:\mathcal{D} \rightarrow \Sigma^*$ such that $\E[|c(\x)|]\leq \ent[\x]/\log_2|\Sigma| +1$.
\end{enumerate}
\end{theorem}

(In fact, this theorem applies to the more general class of \emph{uniquely decodable} codes, but it is sufficient for our purposes that we only consider prefix-free codes.)

This suggests that if we want to upper bound the entropy of $\bigx \sim \widehat{f}^2$, we should try to design an efficient protocol for communicating the value of $\bigx$.  The formula $\Inf[f] = \sum_S |S|\cdot \widehat{f}(S)^2$ shows that $\Inf[f]$ is actually the expected size of the set $\bigx$.  Thus, showing a bound of the form $\ent[\bigx]\leq C \cdot \Inf[f]$ for a function $f$ requires showing a protocol for communicating the value of $\bigx$ which uses at most a constant number of bits on average for each element of $\bigx$.  As an example, consider the following protocol for encoding the value of a set $S \subseteq [n]$:
\begin{center}
\textbf{$\bigp(S)$:}
\end{center}
\begin{itemize}
\item For each $i \in S$, output the $\lceil \log n \rceil$-bit description of $i$.
\item Output $\bot$.
\end{itemize}
Here $\bot$ is a termination character which prevents different codewords from being prefixes of each other. (Without it, the codeword for $\{1\}$ would be a prefix of the codeword for $\{1, 2\}$, for example.)

Given the output of this protocol, one can uniquely determine the value of $S$.  Furthermore, the protocol uses exactly $\lceil \log n \rceil \cdot |S| + 1$ characters to code $S$.  As a result, we have
\begin{equation}
\E[|\bigp(\bigx)|] = \lceil \log n \rceil \cdot \E[|\bigx|] + 1 = \lceil \log n \rceil \cdot \Inf[f] + 1, \label{eq:plusone}
\end{equation}
giving an upper bound of $\ent[\bigx] \leq \log_2 3\cdot \left(\lceil \log n \rceil \cdot \Inf[f] + 1\right)$.  This is (ignoring the $\log_2 3$ factor) the well-known ``weak'' upper bound~\cite{OWZ11, KMS12}, which is essentially the best-known upper bound for a general Boolean $f$ (and is tight when $f$ is real-valued).

With some extra work, we can remove the $(+1)$ from Equation~\eqref{eq:plusone} while adding only a small factor to the coefficient of $\Inf[f]$.  This is important for the case when $f$ is heavily biased and $\Inf[f]$ is small (for example, when $f$ is the AND function).  As a start, consider the modified protocol $\bigp'$ which has the same first line as $\bigp$ but the following second line instead:
\begin{itemize}
\item If $S \neq \emptyset$, output $\bot$.
\end{itemize}
This will only output $\bot$ when $S \neq \emptyset$. For $\bigx \sim \widehat{f}^2$, the probability that $\bigx \neq \emptyset$ is
$\sum_{S \neq \emptyset} \widehat{f}^2(S) = \Var[f] \leq \Inf[f]$.
As a result, $\E[|\bigp'(\bigx)|] \leq (\lceil \log n\rceil + 1) \cdot \Inf[f]$.  However, $\bigp'$ is no longer prefix-free: $\bigp'(\emptyset)$ is the empty string, and is therefore a prefix of $\bigp'(S)$ for \emph{every} $S$.  The following lemma, which is implicit in~\cite{OWZ11}, shows that such a protocol still gives an entropy bound at a cost of $2\cdot \Inf[f].$
%Fortunately, the following lemma, which is implicit in~\cite{OWZ11}, shows that we can work around this problem at a cost of $2 \cdot \Inf[f]$.

\begin{lemma}\label{lem:nozero}
Let $\bigx \sim \widehat{f}^2$, and let $\bigp:2^{[n]} \rightarrow \Sigma^*$ be a prefix-free protocol, except it outputs an empty string on the input $\emptyset$.  Then $\ent[\bigx] \leq \log_2 |\Sigma| \cdot \E[|\bigp(\bigx)|] + 2 \cdot \Inf[f]$.
\end{lemma}\noindent
For completeness, we include a proof of this lemma in Appendix~\ref{sec:nozero}.  Applying this lemma to the protocol $\bigp'$ in the previous example shows that $\ent[\bigx] \leq (\log_2 3 \cdot (\lceil \log n\rceil  +1) + 2)\cdot \Inf[f]$.

As the above example illustrates, it is natural for a protocol to output nothing when $\bigx = \emptyset$.  For convenience, we will call such protocols \emph{almost} prefix-free.\footnote{An almost prefix-free protocol is implicit in the proof of the FEI conjecture for symmetric functions in~\cite{OWZ11}, and ignoring the case when $\bigx = \emptyset$ is even explicitly built into the definition of the FEI$^+$ conjecture in~\cite{OT13}.}

\subsection{Decision Tree Protocol}

\begin{figure}
\centering
\includegraphics{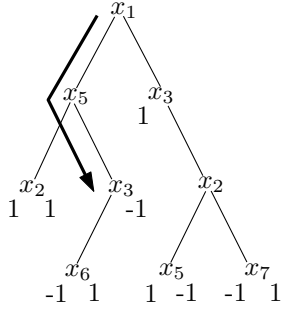}
\caption{A path for the set $S = \{1, 3\}$.  The other possible path is $x_1 \rightarrow x_3$.}\label{fig:path}
\end{figure}

Let $f:\{-1, 1\}^n \rightarrow \{-1, 1\}$ be computed by a decision tree $T$, and let $\bigx \sim \widehat{f}^2$.  To prove Theorems~\ref{thm:readk} and~\ref{thm:depthd}, we give an efficient protocol for communicating the value of $\bigx$.  The protocol we use is simple: for a set $S \subseteq [n]$, $\widehat{f}(S)^2$ can be nonzero only if there is a root-to-leaf path in $T$ which contains all the variables in $S$ and, potentially, some extra variables.  This means that any value which $\bigx$ takes with nonzero probability must correspond to at least one such path in the tree $T$.  The protocol outputs the left/right description of such a path (stopping when the path has reached all the variables in $\bigx$), along with a sequence of bits indicating which indices along the path are contained in~$\bigx$.  Then, if $\bigx \neq \emptyset$, it terminates with a $\bot$.

For example, consider the tree in Figure~\ref{fig:path}.  If the protocol were given the set $S = \{1, 3\}$, then there are two paths it could use: $x_1 \rightarrow x_5 \rightarrow x_3$ and $x_1 \rightarrow x_3$.  Supposing it chose the first path, it would output $0, 1$ for the description of the path, then $1, 0, 1$ to indicate that $x_1$ and $x_3$ are in $S$ but $x_5$ is not, and finally it would output $\bot$.  So the total output string would be $0, 1, 1, 0, 1, \bot$.  If it used the other path, the output string would be $1, 1, 1, \bot$.  We defer the complete description of the protocol, including how it chooses between the possible paths, until Section \ref{sec:readk}.

Note that when $\bigx = \emptyset$, the protocol simply outputs an empty path.  We show the following bound on the performance of this protocol which, when combined with Lemma~\ref{lem:nozero} (and the fact that $k$ and $d$ are at least $1$), yields Theorems~\ref{thm:readk} and~\ref{thm:depthd}:
\begin{theorem}\label{thm:readkprot}
Suppose $f:\{-1, 1\}^n \rightarrow \{-1, 1\}$ is computable by a read-$k$ decision tree whose expected depth is $d$, and let $\bigx \sim \widehat{f}^2$.  Then there is an almost prefix-free protocol for $\bigx$ with length at most $\min\{(2k+2)\cdot \Inf[f], 4\cdot \Inf[f] + 2d\}$ and alphabet size $\vert \Sigma \vert = 3$.
\end{theorem}

This protocol relies heavily on the intuition that the structure of a decision tree should indicate which variables are significant.  For example, the root variable should be very important, as should variables in the upper levels of the tree.
Thus, even though the path outputted by the protocol always includes the root variable and almost always includes the variables in the upper levels of the tree, this should not be a problem given that these variables are highly influential.

It is possible, however, to construct trees which do not fit this intuition: for example, consider a decision tree $T$ which contains one set of variables on levels $0$ through $l-1$, and has rooted at every node on level $l$ a copy of a decision tree $T'$ over a different set of variables.  An example of such a tree is given in Figure~\ref{fig:badtree} for $l = 2$.  As all paths lead to $T'$, the variables in the first $l$ levels clearly have influence zero.  Unfortunately, the described protocol will always output a path containing a variable from each of these $l$ levels, as each path from the root to an influential variable must go through these levels.  Thus, for any arbitrary $l$, one can make this protocol output $2l$ extraneous characters for any nonempty input set, regardless of the influence of the function.

\begin{figure}
\centering
\includegraphics{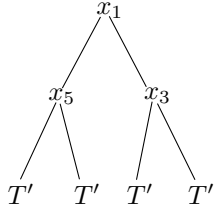}
\caption{A bad tree.  The $T'$s are identical and do not contain $x_1$, $x_3$, or $x_5$.}\label{fig:badtree}
\end{figure}

This is not so problematic for the case when $T$ has small expected depth or is read-$k$, for $k$ a small constant.  In the above example, every level of dummy variables adds one to the depth of $T$, and so this construction is limited by the expected depth of $T$.  Furthermore, since a copy of $T'$ is rooted at every level-$l$ node, $T$ is itself at least a read-$2^l$ decision tree, in which case the fact that the protocol outputs only $2l$ more bits than it should is perhaps not too concerning.

To analyze this example, we note that for each level $i$ between~$0$ and~$l-1$, every node at level $i$ has a pair of highly covariant children.  
Here, by the covariance of two functions $f$ and $g$ we mean the quantity $\Cov[g, h] := \E_{\x}[(g(\x) - \E[g])\cdot(h(\x)-\E[h])]$.
In other words, for a node at level $i$, if $g$ and $h$ are the functions computed by that node's left and right subtrees, respectively, then as $g = h$, $\Cov[g, h] = \Var[f]$.  Imagining that $\Var[f]$ is large, then it is exactly these nodes with highly covariant children which are troublesome.  To keep track of these troublesome nodes, we define the quantity of \emph{tree covariance} for $T$, written $\Cov[T]$.  If $T$'s left and right subtrees $T_0$ and $T_1$ compute the functions $g$ and $h$, then $\Cov[T]$ can be defined recursively as $\Cov[T] = \Cov[g, h] + \frac{1}{2}\left(\Cov[T_0] + \Cov[T_1]\right)$, with the base case that $\Cov[T] = 0$ if $T$ computes a constant function.  We show that the performance of this protocol on a general tree $T$ depends on $\Cov[T]$:

\begin{lemma}\label{thm:covthm}
The length of the above protocol is $4 \cdot \Inf[f] + 2 \cdot \Cov[T]$.
\end{lemma}

It is a simple fact (see Proposition \ref{prop:depthdcov}) that $\Cov[T] \leq d$ if $T$ has expected depth $d$, and so Lemma~\ref{thm:covthm} implies that the length of the protocol is at most $4\cdot \Inf[f] + 2d$, which gives a part of Theorem~\ref{thm:readkprot}.

Upper bounding $\Cov[T]$ for read-$k$ decision trees is more complicated. For intuition, consider the case when $k = 2$.  
Again, suppose $T$'s left and right subtrees $T_0$ and $T_1$ compute the functions $g$ and $h$.
At the extreme, if $\Cov[g, h]$ were to equal one, then this would mean that $g = h$, in which case every variable relevant to $g$ is also relevant to $h$, and vice versa.  In particular, every variable queried in $T_0$ to compute $g$ must also be queried in $T_1$ to compute $h$, meaning that $T_0$ cannot have any variables which appear twice (as $T$ is read-twice).  And if $T_0$ is read-once, then the functions computed by its left and right subtrees must be entirely uncorrelated, as they depend on different variables.  Thus, in this case $\Cov[T_0] = \Cov[T_1] = 0$, so $\Cov[T] = \Cov[g, h] = 1$.  The result, intuitively, is that $T$ has a finite amount of tree covariance to go around, and once it uses it up at a given level, the remaining levels must be uncorrelated.  We extend this intuition into a bound on the tree covariance for read-$k$ decision trees.

\begin{lemma} \label{thm:covbound}
Let $f:\{-1, 1\}^n \rightarrow \{-1, 1\}$ be computed by a read-$k$ decision tree $T$.  Then $\Cov[T] \leq (k-1) \cdot \Var[f]$.
\end{lemma}
Combining this lemma with Lemma~\ref{thm:covthm} and the fact that $\Var[f] \leq \Inf[f]$ shows that the length of the protocol is at most $(2k+2)\cdot \Inf[f]$, giving the remaining part of Theorem~\ref{thm:readkprot}.

\subsection{Read-Once Composition Protocol}
Theorem \ref{thm:informal} from \cite{OT13} shows that composing functions which satisfy the FEI$^+$ conjecture will result in a function which also satisfies FEI$^+$.  We give a new proof of this theorem by proving an analogous result (Theorem \ref{thm:protocolcomp} below) for \textit{protocols} instead of entropy; our proof shows how to construct an efficient protocol for the composed function using the efficient protocols of each of the functions in the composition.  To complete the proof of Theorem \ref{thm:informal}, which is a statement about entropies, one might try to use the source coding theorem to translate our result about protocols to a result about entropies.  This can't be done so simply, however, as Theorem~\ref{thm:shannon} only gives an approximate correspondence between protocols and entropy. We are able to get this step to work by using a (mostly standard) parallelizing technique.  We now describe each of these two steps in more detail.
%\jnote{Maybe this is underselling us?})

%However, Theorem~\ref{thm:shannon} only gives an approximate correspondance between protocols and entropy, whereas we need something stronger. In spite of this, we are able

%obtain good protocols from the entropy bounds. However, as noted after Theorem \ref{thm:shannon}, the encodings may incur an additive blow-up from the entropy bound.  We show how to modify the composed protocol in the proof of Theorem \ref{thm:protocolcomp} to accommodate the additive blow-up from the inner protocols.  We now describe each of these two steps in more detail.

The FEI$^+$ conjecture works with the spectral distribution conditioned on the sample being non-empty. We write this distribution as
%For convenience and compatibility with \cite{OT13},
%our composition theorem will consider
%As a result, it will be useful to define
%the random variable
%the distribution $\widehat{f}^2\setminus\emptyset$, which is the distribution $\widehat{f}^2$ conditioned on it not outputting the empty set.  If $\bigy$ is a random variable distributed according to $\widehat{f}^2\setminus \emptyset$, which we write as
$\bigy \sim \widehat{f}^2 \setminus \emptyset,$ which is defined so that:
\begin{equation*}
\Pr[\bigy = \emptyset] = 0,
\quad\text{and}\quad
\Pr[\bigy = S] = \frac{\widehat{f}(S)^2}{1 - \widehat{f}(\emptyset)^2},
\end{equation*}
for any $S \neq \emptyset$.  %Of course, this distribution is sensible only if
We assume here that $\widehat{f}(\emptyset)^2 < 1$ (the FEI Conjecture is trivial when $\widehat{f}(\emptyset)^2 = 1$).
For our purposes, a prefix-free protocol $\bigp$ for $\bigy$ is the same as an almost prefix-free protocol for $\bigx \sim \widehat{f}^2$:  the equality $\E[|\bigp(\bigx)|] = \Var[f]\cdot\E[|\bigp(\bigy)|]$ holds, and Lemma \ref{lem:nozero} tells us that we may obtain a bound on the entropy of $\bigx$ using a prefix-free protocol for $\bigy$.
%In addition, Lemma~\ref{lem:nozero} could be rewritten as $\ent[\bigx] \leq \Var[f] \cdot \ent[\bigy] + 2\cdot \Inf[f]$, though we will not need this reformulation.\anote{Do we mean $\E[ |P(Y)|]$ here? Maybe we should just omit this line\ldots}
%As these two viewpoints are equivalent up to normalization, we will choose our approach based on convenience.  Our read-$k$ decision tree upper bound will follow from an almost prefix-free protocol for $\bigx$, and our new proof of~\cite{OT13}'s composition theorem will follow from a prefix-fee protocol for $\bigy$.

The FEI$^+$ conjecture in \cite{OT13} strengthens the FEI conjecture and generalizes it to product distributions, making it amenable to composition.   We use $\widetilde{f}$ to denote the Fourier transform of $f$ with respect to a product distribution $\mu$ (here each bit $x_i$ is set so that $\E_\mu[x_i]=\mu_i$).  We now state the main definition from $\cite{OT13}$:
\begin{definition}\label{def:feiplus}
Let $f:\{-1,1\}^n_\mu \to \{-1,1\}$ be a Boolean function. The function $f$ satisfies FEI$^+$ with constant $C$ if
$$  \sum_{S\neq \emptyset} \widetilde{f}(S)^2\log \left(\frac{\prod_{i\in S} (1-\mu^2_i)}{\widetilde{f}(S)^2}
\right)\leq C \cdot \sum_{S\neq \emptyset} \widetilde{f}(S)^2 (|S|-1).$$

\end{definition}
\noindent In~\cite{OT13}, it was conjectured that for some constant $C$, every Boolean function satisfies FEI$^+$ with constant $C$.  They were in fact able to show that every Boolean function $f$ satisfies FEI$^+$ with ``constant" $2^{O(n)}$.\footnote{It is known that one can improve this to $O(\log(n))$ in the unbiased case when all the $\mu_i$'s are zero.}
%\jnote{Should we mention this?  I feel like we spent enough of our time on this that this warrants a mention...}

%Recall that \cite{OT13} works with a stronger version of the FEI conjecture that generalizes to product distributions and is amenable to composition, the FEI$^+$ conjecture.
%As a first step,
Our first step is to reformulate what it means to ``satisfy the FEI$^+$ conjecture with constant $C$"
%(see after Conjecture 1 in \cite{OT13})
as a statement about the existence of an efficient protocol:
\begin{definition}
Let $f:\{-1, 1\}_\mu^n \rightarrow \{-1, 1\}$ be a function over the $\mu$-biased variables $x_1, \ldots, x_n$, and let $\bigy \sim \widetilde{f}^2\setminus \emptyset$.  Let $P$ be a prefix-free protocol for communicating the value of $\bigy$.  Then $P$ is a \emph{$C$-good} protocol for $f$ under bias $\mu$ if
\begin{equation*}
\E[|P(\bigy)|] \leq C\cdot\left( \E[|\bigy|] - 1\right) + \sum_{i} \Pr[i \in \bigy] \cdot \log \frac{1}{1-\mu_i^2} + \log \Var_p[f].
\end{equation*}
\end{definition}
This definition can be derived by rearranging the inequality in Definition \ref{def:feiplus} to place \newline $\sum_{S\neq \emptyset} \widetilde{f}(S)^2\log \frac{1}{\widetilde{f}(S)^2}$ on the left-hand side, and then replacing $\sum_{S\neq \emptyset} \widetilde{f}(S)^2 \log \frac{1}{\widetilde{f}(S)^2} = \ent[\bigy]$ with $\E[|P(\bigy)|]$.
%\jnote{I feel like we need a statement somehwere of how we arrive at this definition.  I don't know if what I just wrote is sufficiently understandable though}
Because $\ent[\bigy] \leq \E[|P(\bigy)|]$, any function with a good protocol automatically satisfies FEI$^+$:
\begin{fact}\label{fact:entropyprotocol}
Suppose there exists a $C$-good protocol for $f$ under bias $\mu$.  Then $f$ satisfies (the $\mu$-biased) FEI$^+$ with constant $C$.
\end{fact}

We then prove the following composition theorem for protocols in  Section \ref{sec:composition}:
\begin{theorem}\label{thm:protocolcomp}
Let $h(x^1, \ldots, x^k) = f(g_1(x^1), \ldots, g_k(x^k))$, where the domain of $h$ is endowed with a product distribution $\mu$.  Suppose there are $C$-good protocols for $g_1,\ldots,g_k$ under $\mu$ and a $C$-good protocol for $f$ under bias $\eta = \langle \E_\mu[g_1],\ldots,\E_\mu[g_k]\rangle$. Then there exists a $C$-good protocol for $h$ under bias $\mu$.
\end{theorem}
Given a good protocol $P_f$ for $f$ and good protocols $P_1,\ldots,P_k$  for $g_1,\ldots, g_k$, we construct a good protocol for $h$ in the following way. Let $\bigy = \bigy_1 \circ \cdots \circ \bigy_k$ be drawn from $\widetilde{h}^2\setminus\emptyset,$ where each
$\bigy_i$ denotes the restriction of $\bigy$ to the relevant coordinates of $g_i$. Note that the $\bigy_i$'s form a partition of $\bigy$ because the $g_i$'s have disjoint inputs.
The protocol will use $P_f$ to specify which $\bigy_i$ are non-empty, and, for each such $i$, it will use $P_i(\bigy_i)$ to specify which of the bits relevant to $g_i$ are present in $\bigy$.  While outputting all of $P_1(\bigy_1),\cdots,P_k(\bigy_k)$ would be simpler and would suffice to completely specify $\bigy$,  this protocol will not be efficient when the  $g_i$'s have small variance (in this case the number of non-empty $\bigy_i$ may be quite small).

In fact, the set $S\subseteq [k]$ of non-empty $\bigy$ will be distributed according to $\widetilde{f}^2 \setminus \emptyset$, where $\widetilde{f}$ denotes the $\eta$-biased Fourier transformation of $f$, and furthermore, the sets $\bigy_i$ are distributed according to $\widetilde{g_i}^2\setminus \emptyset.$  This fact is somewhat implicit in the analysis of \cite{OT13}, though we find it somewhat clearer and simpler to prove in isolation, without reference to entropy.
The analysis of this protocol follows almost immediately from this fact, as the protocols $P_f$ and $P_1,\ldots,P_k$ are designed for these distributions.

This yields a composition theorem for protocols.  Our ultimate goal, however, is to prove the following composition theorem for FEI$^+$:
\begin{theorem}\label{thm:entropycomp}
Let $h(x^1, \ldots, x^k) = f(g_1(x^1), \ldots, g_k(x^k))$, where the domain of $h$ is endowed with a product distribution $\mu$.  Suppose $g_1, \ldots, g_k$ satisfy $\mu$-biased FEI$^+$ with constant $C$ and $f$ satisfies $\eta$-biased FEI$^+$ with constant $C$, where $\eta = \langle \E_\mu[g_1],\ldots,\E_\mu[g_k]\rangle$. Then $h$ satisfies $\mu$-biased FEI$^+$ with constant $C$.
\end{theorem}
The naive strategy would be to apply Shannon's source coding theorem to derive $C$-good protocols for $f, g_1, \ldots, g_k$, apply Theorem~\ref{thm:protocolcomp} to give a $C$-good protocol for $h$, and then apply Fact~\ref{fact:entropyprotocol} to show that $h$ satisfies FEI$^+$.  Unfortunately, this fails in the first step: the source coding theorem loses an additive factor of $(+1)$ when translating from entropy to protocols, and this $(+1)$ means that $f, g_1, \ldots, g_k$ don't necessarily have $C$-good protocols.

To fix this problem, we use the well-known observation that the length of a protocol can be made arbitrarily close to the entropy of a given random variable by encoding many independent copies of that random variable.  Thus, by switching to protocols which encode multiple copies of $\bigy$ instead of just one, we can ensure that the first step goes through properly, and the other steps (such as Theorem~\ref{thm:protocolcomp}) go through nearly identically in this setting as well. As this part of the argument is essentially standard, we sketch it briefly in Appendix \ref{sec:parallel}.

%in order to recover the composition theorem in \cite{OT13} about entropy, %some further steps are necessary.
%another step is necessary.
%One slight nuisance here is in translating Theorem~\ref{thm:protocolcomp} from a statement about protocols to a statement about entropy.
%We would like to show that if $f$ and the $g_i$'s satisfy the FEI$^+$ conjecture with constant $C$, then so does $h$.  The natural way to do this would be to first apply Shannon's source coding theorem to get the $C$-good protocols for $f$ and the $g_i$'s, then to show that these protocols compose nicely, and finally to apply the other direction of the source coding theorem to imply that $h$ satisfies the FEI$^+$ conjecture with constant $C$.  Unfortunately, this breaks down in the first step:\anote{hmm, we might consider putting a bit about this after the statement of Shannon's theorem, though we will still need something here.}
\subsection{Organization}
The decision tree results can be found in Section~\ref{sec:readk} , and the FEI$^+$ results can be found in Section~\ref{sec:composition}.  The appendices mostly contain proofs of simple lemmas.  Appendix~\ref{app:biginf} contains the argument for why the restriction on the total influence of $f$ in Theorem~\ref{thm:depthd} is necessary.

\paragraph{Proofs of the main theorems.} Theorem~\ref{thm:readk} and Theorem~\ref{thm:depthd} follow from Lemma~\ref{lem:nozero} and Theorem~\ref{thm:readkprot}. %where Theorem~\ref{thm:readkprot} follows from Lemma~\ref{thm:covthm}, Lemma~\ref{thm:covbound}(proved in Section ~\ref{sec:readkcov}) and Proposition~\ref{prop:depthdcov}. 
Theorem~\ref{thm:informal} follows from Theorem~\ref{thm:protocolcomp} (proved in Section~\ref{sec:composition}) and Theorem~\ref{thm:entropycomp}.

\section{Entropy-Influence for read-$k$ decision trees}\label{sec:readk}
In this section, we analyze our communication protocol for decision trees .    We begin with some preliminary definitions in Section~\ref{sec:defs}.  Then, as a simple first step, we consider the case of read-once decision trees in Section~\ref{sec:readonce}.  Finally, we prove Lemma~\ref{thm:covthm} in Section~\ref{sec:readk-prot} and Lemma~\ref{thm:covbound} in Section~\ref{sec:readkcov}. Together, these prove Theorem \ref{thm:readkprot}.

\subsection{Definitions and Notation}\label{sec:defs}

\paragraph{Fourier analysis.}
Unless stated otherwise, a random input $\x \in \{-1, 1\}^n$ has the uniform distribution.
Any function $f:\{-1, 1\}^n \rightarrow \mathbb{R}$ can be written as
\begin{equation*}
f(x) = \sum_{S \subseteq [n]}\widehat{f}(S) \chi_S(x).
\end{equation*}
The $\widehat{f}(S)$'s are the \emph{Fourier coefficients} of $f$, and for each $S\subseteq [n]$, the parity function $\chi_S$ is defined as $\chi_S(x) = \prod_{i \in S} x_i$. \emph{Parseval's equation} will be important for us, which states that $\E_\x[f(\x)^2] = \sum_S \widehat{f}(S)^2$.  In particular, if $f$ is $\pm 1$-valued, then this sum equals one, and so the squared coefficients $\widehat{f}(S)^2$ form a probability distribution.  We will also need the formula $\Var[f] = \sum_{S \neq \emptyset} \widehat{f}(S)^2$.  We note that if $\bigx \sim \widehat{f}^2$, then $\Pr[\bigx \neq \emptyset] = \Var[f]$.

The \emph{influence} of a variable $x_i$ on $f$ is $\Inf_i[f] := \Pr_\bx[f(\bx)\neq f(\bx^{\oplus i})]$, where $\bx^{\oplus i}$ is $\bx$ with the $i$-th bit flipped.  The \emph{total influence} of $f$ is $\Inf[f] := \sum_i \Inf_i[f]$, and it is simple to show that $\Inf[f]$ can also be written as $\Inf[f] = \sum_{S \neq \emptyset} |S| \widehat{f}(S)^2$.  Comparing this to the formula for $\Var[f]$ shows that $\Var[f] \leq \Inf[f]$.  This is all the Fourier analysis we will need; for a more comprehensive introduction to the subject, see~\cite{OD13}.

\paragraph{Decision trees.}
Decision trees are a standard model of computation, and we omit their definition  (see, for example, \cite{OWZ11} for a definition).  Given a tree $T$, we will call the subtree corresponding to the $+1$ edge the \emph{left} subtree and the subtree corresponding to the $-1$ edge the \emph{right} subtree. We will assume that if $T$ is a decision tree, then no variable appears more than once in any root-to-leaf path of $T$.  If this is not the case, then $T$ can be simplified.  We say that $T$ is a \emph{read-$k$} decision tree if no variable is queried in more than $k$ locations of $T$.

Given a decision tree $T$, if $v$ is a node of $T$, then $l(v)$ is the \emph{label} of $v$, i.e. the coordinate in $x$ which is queried at node $v$.  Let $r(T)$ be the root node of $T$.  Next, set $d(v)$ to be the depth of $v$ in $T$.  We start counting the depth at~$0$, so that $d(r(T)) = 0$.  The \emph{expected depth} of $T$ is the average number of bits $T$ queries on a uniformly random input $\bx$.  Since a given node $v$ is reached with probability $2^{-d(v)}$, the expected depth of $T$ may be written as
\begin{equation}
\sum_{v \in T} 2^{-d(v)}.\label{eq:expecteddepth}
\end{equation}

\ignore{Given a decision tree $T$ and a path $P = p_1, \ldots, p_k$ in the tree (starting at the root $p_1$), define the \emph{description} of the path to be the sequence of bits $b_1, \ldots, b_k$ which, if read in that order, would result in traversing the given path.
Let $r(T)$ be the root of $T$.  If $v$ is an internal node of $T$, then $l(v)$ is the \emph{label} of $v$, i.e. the coordinate in $x$ which is queried at node $v$.  Next, set $d(v)$ to be the depth of $v$ in $T$.  We start counting the depth at~$0$, so that $d(r(T)) = 0$.  The depth of a tree is largest depth of any of its leaves. The expected depth of a tree is defined to be the expected depth of a random leaf (chosen by a uniform random input to the tree). Since each leaf is reached with probability $2^{-d(\ell)}$, the expected depth of $T$ may be written as:
$$\sum_{\ell \in \text{leaves}(T)} 2^{-d(\ell)} d(\ell).$$}

Given two functions $g,h:\{-1, 1\}^n \rightarrow \mathbb{R}$, define $\Cov[g, h] := \E_{\x}[(g(\x) - \E[g])\cdot(h(\x)-\E[h])]$.  Now we may state our main definition:
\begin{definition}
Given a decision tree $T$ and an internal node $v$, let $g$ be the function computed by $v$'s left subtree and $h$ be the function computed by $v$'s right subtree.  Then define
\begin{itemize}
\item $\Cov[v] := \Cov[g, h]$,
\item $\Cov_i[T] := \sum_{v: l(v) = i} \Cov[v] \cdot 2^{-d(v)}$, and
\item $\Cov[T] := \sum_{v \in T} \Cov[v] \cdot 2^{-d(v)}$.
\end{itemize}
\end{definition}
Note that $\Cov[T]$ may also be written as $\Cov[T] = \sum_{i \in [n]} \Cov_i[T]$.  Furthermore, if $T_0$ is $T$'s left subtree and $T_1$ is $T$'s right subtree, then $\Cov[T]$ may also be written recursively as $\Cov[T] = \Cov[g, h] + \frac{1}{2} \left(\Cov[T_0] + \Cov[T_1]\right)$, with the base case that $\Cov[T] = 0$ if $T$ performs no queries.
Intuitively, $\Cov[T]$ is a measure of the total correlation present in the structure of $T$. For example, $\Cov[T] = 0$ if $T$ is a read-once decision tree.  We note that when $T$ computes a Boolean function,  $\Cov[v] \leq 1$ for each $v \in T$.  Thus, in this case, it is immediate from Equation~\eqref{eq:expecteddepth} that the expected depth of $T$ is at least $\Cov[T]$.  This gives the following proposition.
\begin{proposition}\label{prop:depthdcov}
Let $f:\{-1, 1\}^n \rightarrow \{-1, 1\}$ be computed by $T$, a decision tree with expected depth $d$. Then $\Cov[T] \leq d$.
\end{proposition}

We will also need the following two propositions, which are proven in Appendix \ref{sec:readk-prelim}.
\begin{proposition}\label{prop:commonfact}
Let $f$ be computed by a decision tree $T$ whose left and right subfunctions are $g$ and $h$, respectively.  If $x_i$ is at the root of $T$ and $S$ is any subset of $[n] \setminus \{i\}$, then
\begin{equation*}
\widehat{f}(S)^2 + \widehat{f}(S \cup \{i\})^2 = \frac{1}{2}\left(\widehat{g}(S)^2 + \widehat{h}(S)^2\right).
\end{equation*}
\end{proposition}
\begin{proposition}\label{prop:obviousinfluencefact}
Assume the setup from Proposition~\ref{prop:commonfact}.  Then for all coordinates $j \neq i$,
\begin{equation*}
\Inf_j[f] = \frac{1}{2}\cdot(\Inf_j[g] + \Inf_j[h]).
\end{equation*}
\end{proposition}
%\subsection{Organization}
%In Section~\ref{sec:readonce}, we give a sketch of our proof in the case when the decision tree is read-once.  In Section~\ref{sec:readk}, we generalize this proof to show that the spectral entropy of a function computed by a decision tree $T$ is bounded by an expression involving the covariance of $T$.  Finally, in section~\ref{sec:readkcov} we show that $\Cov[T]$ is small when $T$ is read-$k$.

\subsection{Read-once decision trees}\label{sec:readonce}
In this section, we will sketch the argument for read-once decision trees.  Let $f:\{-1, 1\}^n \rightarrow \{-1, 1\}$ be computed by a read-once decision tree $T$.  Given a decision tree $T$ and a path $P = v_1\rightarrow \cdots\rightarrow v_k$ in the tree (starting at the root $v_1$),  the \emph{description} of the path is the sequence of bits $b_1, \ldots, b_{k-1} \in \{0, 1\}$ which, if read in that order, would result in traversing the given path (here we are using the standard $1\leftrightarrow 0$ and $-1\leftrightarrow 1$ correspondance). Given a set $S$ for which $\widehat{f}(S) \neq 0$, our protocol will output the description of a path in which $S$ is a subset of $\{l(v_1), \ldots, l(v_k)\}$. In fact, our protocol will choose a \emph{minimal} such path containing $S$, in the sense that the path will stop once it has encountered all of the variables in $S$.  In a general decision tree, there could be many minimal paths containing $S$ and starting at the root, but because $T$ is read-once, there can only be one such path.  We may therefore state the protocol as:

\ignore{
\begin{center}
\textbf{Given $S \subseteq [n]$:}
\end{center}
\begin{itemize}
\item If $S = \emptyset$, output nothing.
\item There is a minimal path $P = p_1\rightarrow \cdots\rightarrow p_k$  containing the indices in $S$ which starts at $T$'s root.
\item Output the description of that path.
\item Output a bit sequence $b_1, \ldots, b_k \in \{0, 1\}^k$, where $b_i = 1$ iff $p_i \in \bigx$.
\item Terminate with a $\bot$.
\end{itemize}
}

\medskip
\fbox{\parbox{14.5cm}{
\begin{center}
\textbf{Given $S \subseteq [n]$:}
\end{center}
\begin{enumerate}
\item If $S = \emptyset$, output nothing.
\item There is a minimal path $P = p_1\rightarrow \cdots\rightarrow p_k$  containing the indices in $S$ which starts at $T$'s root.
\item Output the description of that path.
\item Output a bit sequence $b_1, \ldots, b_k \in \{0, 1\}$, where $b_i = 1$ iff $p_i \in \bigx$.
\item Terminate with a $\bot$.
\end{enumerate}
}}
\medskip

\noindent
We stress that the protocol is only required to work properly when $S$ corresponds to a nonzero Fourier coefficient, i.e.\ $\widehat{f}(S) \neq 0$.

Suppose that the path $P$ the protocol finds is of length $l$.  Then because the description of a path of length $l$ uses $l-1$ bits, the protocol outputs $2l$ characters in total.  Furthermore, the protocol accurately communicates the value of $S$: given the output of the protocol, one could reconstruct $S$ by following the path indicated by the first $l-1$ bits and including only those indices along the path which are tagged with a~$1$ in the second sequence.  So long as $S \neq \emptyset$, the output is terminated with a $\bot$ character.  Together, these mean that the protocol is an almost prefix-free protocol with alphabet size $|\Sigma| = 3$.

We are interested in the length of the protocol on input $\bigx \sim \widehat{f}^2$.  As shown above, the number of characters this protocol outputs is exactly twice the length of the path $P$.   Thus, we need to upper bound the average length of $P$.

Let us consider reasons why $P$ might be on average too long.   For example, because the protocol only considers paths starting at the root, the path output always contains the root variable (unless $\bigx = \emptyset$), even though this variable might have very low influence on $f$.  However, a simple argument shows that this worry is unfounded.  In particular, if $x_i$ is $T$'s root variable, then $\Inf_i[f] \geq \frac{1}{2}\Var[f]$ (we will show this later in Lemma~\ref{lem:inf}).  This inequality uses crucially the fact that $T$ is read-once.  The path $P$ contains $x_i$ whenever $\bigx \neq \emptyset$, which happens with probability $\Var[f]$.  Thus, the probability that $P$ contains $x_i$ is at most~$2\cdot\Inf_i[f]$.

An inductive argument allows us to bring this inequality down to the rest of the variables in the tree, showing that the probability $P$ contains a variable $x_j$ is at most $2 \cdot \Inf_j[f]$ (we will show this later in Lemma~\ref{lem:main}).  Summing this inequality over all $j$ shows that the expected length of $P$ is at most $2\cdot\Inf[f]$.  Thus, the protocol outputs at most $4\cdot\Inf[f]$ characters in expectation, proving Theorem~\ref{thm:covthm} in the $k=1$ case.

\subsection{General decision trees}\label{sec:readk-prot}

Let $f:\{-1, 1\}^n \rightarrow \{-1, 1\}$ be computed by a decision tree $T$.  Generalizing the above argument to work for $T$ requires some modifications.  The main change is that given $S \subseteq [n]$, there is no longer necessarily a unique minimal path starting from $T$'s root which contains the indices in $S$.  As Figure~\ref{fig:path} shows, there could be two paths to select from when, for example, $S = \{1, 3\}$.  We want our protocol to use the fewest characters possible, so the obvious choice is for it to simply use the shortest path possible.  This protocol is difficult to analyze, however, so we instead use a suboptimal protocol which constructs a path vertex-by-vertex probabilistically.  If $g$ is the function computed by $T$'s left subtree and $h$ is the function computed by $T$'s right subtree, then the first step of the path will be chosen based on the relative weight that $g$ and $h$ place on the set $S$, i.e. $\widehat{g}(S)^2$ versus $\widehat{h}(S)^2$.  As a result, the protocol is most easily stated recursively, as follows:

\ignore{
\begin{center}
\textbf{$\bigp(T, S)$:}
\end{center}
\begin{itemize}
\item If $S = \emptyset$, output nothing and terminate.
\item Let $g$ be the function computed by $T$'s left subtree $T_0$, and likewise let $h$ be the function computed by $T$'s right subtree $T_1$.
\item Let $x_i$ be $T$'s root variable.  If $i \in S$, output a~$1$.  Otherwise, output a~$0$.
\item Set $S' = S\setminus \{i\}$.  If $S' = \emptyset$, output $\bot$ and terminate.
\item With probability proportional to $\widehat{g}(S')^2$, output~$0$ and run $\bigp(T_0, S')$.
\item With probability proportional to $\widehat{h}(S')^2$, output~$1$ and run $\bigp(T_1, S')$.
\end{itemize}
}

\medskip
\fbox{\parbox{14.5cm}{
\begin{center}
\textbf{$\bigp(T, S)$:}
\end{center}
\begin{enumerate}
\item If $S = \emptyset$, output nothing and terminate.
\item Let $g$ be the function computed by $T$'s left subtree $T_0$, and likewise let $h$ be the function computed by $T$'s right subtree $T_1$.
\item Let $x_i$ be $T$'s root variable.  If $i \in S$, output a~$1$.  Otherwise, output a~$0$.
\item Set $S' = S\setminus \{i\}$.  If $S' = \emptyset$, output $\bot$ and terminate.
\item With probability proportional to $\widehat{g}(S')^2$, output~$0$ and run $\bigp(T_0, S')$.
\item With probability proportional to $\widehat{h}(S')^2$, output~$1$ and run $\bigp(T_1, S')$.
\end{enumerate}
}}
\medskip

This protocol outputs the same information that the protocol in Section~\ref{sec:readonce} does, only now the description of the path and the bit sequence are interleaved.  If this protocol outputs $2k$ characters, then characters~$2$,~$4$,~$\ldots$~,~$2k-2$ give a description of a path $P$,  characters~$1$,~$3$,~$\ldots$~,~$2k-1$ indicate which indices along the path $P$ are included in $S$, and the $2k$-th character is a $\bot$. As a result, this protocol is an almost prefix-free protocol with alphabet size $|\Sigma|=3$.  We will  refer to the path $P$ as the path the protocol outputs, selects, etc.

Let us now consider the length of the protocol on input $\bigx \sim \widehat{f}^2$.  The number of characters output is exactly twice the length of the path $P$ the protocol outputs.   Thus, we would like to upper bound the expected length of $P$.  Our main lemma will show that for a given variable $x_i$, we can upper-bound the probability that it appears in $P$ as follows:
\begin{lemma}\label{lem:main}
Let $f:\{-1, 1\}^n \rightarrow \{-1, 1\}$ be computed by a decision tree $T$, let $\bigx \sim \widehat{f}^2$, and let $p_i(T)$ be the probability that the path selected by $\bigp(T, \bigx)$ contains index $i$.  Then $p_i(T) \leq 2 \cdot \Inf_i[f] + \Cov_i[T]$.
\end{lemma}
By summing this lemma over $i \in [n]$, the expected length of $P$ is at most $2\cdot\Inf[f] + \Cov[T]$, and so the expected number of characters output by the protocol is at most $4\cdot \Inf[f]+2\cdot\Cov[T]$, which proves Lemma \ref{thm:covthm}.
%This yields a bound on the expected number of characters the protocol outputs, which is our Theorem~\ref{thm:covthmrestated}:

%\begin{theorem}[Theorem~\ref{thm:covthm} restated]\label{thm:covthmrestated}

%Suppose $f:\{-1, 1\}^n \rightarrow \{-1, 1\}$ is computable by a decision tree $T$, and let $\bigx \sim \widehat{f}^2$.  Then there is an almost prefix-free protocol which communicates the value of $\bigx$ using $4\cdot\Inf[f] + 2\cdot\Cov[T]$ characters in expectation and has alphabet size $|\Sigma| = 3$.

%\end{theorem}

In the special case when $T$ has expected depth $d$, Proposition~\ref{prop:depthdcov} tells us that $\Cov[T] \leq d$, so the protocol uses at most $4\cdot\Inf[f] + 2\cdot d$ characters in expectation, the bound given in Theorem~\ref{thm:readkprot}.  If we further assume that $\Inf[f] \geq 1$, then this quantity is at most $6d\cdot\Inf[f]$. Combining this with Lemma~\ref{lem:nozero} yields
our FEI bound for decision trees of expected depth $d$,
Theorem \ref{thm:depthd}.
%\begin{theorem}[Theorem~\ref{thm:depthd} restated]
%Suppose $f:\{-1, 1\}^n\rightarrow\{-1, 1\}$ is computable by a decision tree whose expected depth is $d$.  Further, suppose $\Inf[f] \geq 1$.  Then $\ent[\widehat{f}^2] \leq (2 + 6 \log_2 3\cdot d )\cdot \Inf[f]$.
%\end{theorem}
In Appendix~\ref{app:biginf}, we argue that proving this theorem without the restriction that $\Inf[f]\geq1$ is unlikely so long as the Fourier Entropy-Influence conjecture remains unproven.  Next, as upper-bounding $\Cov[T]$ is more involved if $T$ is read-$k$, we will defer the proof of the FEI conjecture for read-$k$ decision trees to Section~\ref{sec:readkcov}.
%\subsection{Proofs}

Now we prove Lemma~\ref{lem:main}.  In Section~\ref{sec:readonce}, we stated that if $x_i$ is $T$'s root variable, then $\Inf_i[f] \geq \frac{1}{2}\Var[f]$, supposing that $T$ is read-once.  Unfortunately, this is not true for general (or even read-twice) decision trees.  For example, the root variable could have two identical subtrees as its children, in which case it has influence zero.  For this to happen, though, it must be the case that the two subfunctions have high covariance.

\begin{lemma}\label{lem:inf}
Let $f : \{-1, 1\}^n \rightarrow \{-1, 1\}$ be computed by a decision tree $T$.  If $x_i$ is at the root of $T$, then $\Inf_i[f] \geq \frac{1}{2}\Var[f]- \frac{1}{2}\Cov[r(T)]$.
\end{lemma}
\begin{proof}
Let $g$ be the function computed by $T$'s left subtree and $h$ be the function computed by $T$'s right subtree, so that $f(x) = g(x)$ if $x_i = 1$, and $f(x) = h(x)$ if $x_i = -1$. Then $ \hat{f}(\emptyset)^2 = \left(\frac{\hat{g}(\emptyset) + \hat{h}(\emptyset)}{2}\right)^2$.
As a result,
\begin{align*}
\Inf_i[f]
&= \Pr[f(\x) \neq f(\x^{\oplus i})] \\
&= \Pr[g(\x) \neq h(\x)] \tag{because $g$ and $h$ don't depend on $x_i$}\\
& = \frac{1}{2} - \frac{1}{2} \E[g(\x) h(\x)]\\
& = \frac{1}{2} - \frac{1}{2} \hat{g}(\emptyset) \hat{h}(\emptyset)- \frac{1}{2}\Cov[g, h]\\
& \geq \frac{1}{2} - \frac{1}{2}\left(\frac{\hat{g}(\emptyset) + \hat{h}(\emptyset)}{2}\right)^2 - \frac{1}{2}\Cov[g, h]\tag{using $ab \leq \left(\frac{a + b}{2}\right)^2$}\\
&= \frac{1}{2}\left(1 - \hat{f}(\emptyset)^2\right) - \frac{1}{2}\Cov[g, h]\\
& = \frac{1}{2}\Var[f]- \frac{1}{2}\Cov[g, h].
\end{align*}
Because $\Cov[r(T)] = \Cov[g, h]$, this proves the lemma.
\end{proof}

We now use this to prove Lemma~\ref{lem:main}:
\begin{proof}[Proof of Lemma~\ref{lem:main}]
We prove this by structural induction on the tree $T$, based on whether $x_i$ is at the root of $T$. The lemma is clearly true if $i$ doesn't appear in $T$, so we will assume that it does.
\paragraph{Base case:} In this case, the root of $T$ is $x_i$.  By the protocol above, $x_i$ will \emph{always} be on the path $P$ unless $\bigx = \emptyset$, i.e. the path is empty.  Thus, the probability that $x_i$ is outputted is $1 - \hat{f}(\emptyset)^2 = \Var[f]$.  By Lemma~\ref{lem:inf}, we have that $2 \cdot \Inf[f] \geq \Var[f] - \Cov[r(T)] = p_i(T) - \Cov[r(T)]$.  Because $x_i$ is at the root, it can appear nowhere else in $T$. This means that $\Cov_i[T] = \Cov[r(T)]$, which concludes the base case.

\paragraph{Inductive step:}

In this case, the root of $T$ is not $x_i$, meaning that $x_i$ is queried in one (or both) of $T$'s subtrees.  Let $T_0$ be the left subtree of $T$ and $T_1$ its right subtree, and assume without loss of generality that the root of $T$ is $x_n$.  We will show the following pair of simple equalities:
\begin{enumerate}
\item $\Inf_i[f] = \frac{1}{2}\cdot \left(\Inf_i[g] +\Inf_i[h]\right)$, and \label{item:inf}
\item $p_i(T) = \frac{1}{2} \cdot \left(p_i(T_0) + p_i(T_1)\right)$. \label{item:probs}
\end{enumerate}
Equality~\ref{item:inf} follows directly from Proposition~\ref{prop:obviousinfluencefact}.  Before proving Equality~\ref{item:probs}, let's see how they imply the lemma.
\begin{align}
2\cdot \Inf_i[f]
&= \Inf_i[g] + \Inf_i[h] \nonumber\\
&\geq \frac{1}{2}\left(p_i(T_0) + p_i(T_1) - \Cov_i[T_0] - \Cov_i[T_1]\right)\nonumber\\
& = p_i(T) - \frac{1}{2}\left(\Cov_i[T_0] + \Cov_i[T_1]\right),\label{eq:applyinduction}
\end{align}
where the second line follows from applying the inductive hypothesis to $g$ and $h$.  Since each vertex $v$ in $T_0$ (or $T_1$) is one edge farther from the root in $T$ than it is in $T_0$ (or $T_1$), we get that
$\Cov_i[T] = \frac{1}{2} \Cov_i[T_0] + \frac{1}{2} \Cov_i[T_1]$.  Note that $\Cov[r(T)]$ doesn't contribute anything to $\Cov_i[T]$ because $x_i$ is not at the root of $T$.  Plugging this equality into Equation~\eqref{eq:applyinduction} yields the lemma.

Now, we prove Equality~\ref{item:probs}.  It will be convenient for us to define the modified protocol $\bigp'$:

\medskip
\fbox{\parbox{14.5cm}{
\begin{center}
\textbf{$\bigp'(T, S)$:}
\end{center}
\begin{enumerate}
\item If $S \neq \emptyset$ and $S \neq \{j\}$, where $x_j$ is $T$'s root variable, then run $\bigp(T, S)$.
\item Otherwise:
\begin{enumerate}
\item If $S = \{j\}$, output the characters~$1,\bot$.
\item With probability proportional to $\widehat{g}(\emptyset)^2$,  run $\bigp(T_0, \emptyset)$.
\item With probability proportional to $\widehat{h}(\emptyset)^2$,  run $\bigp(T_1, \emptyset)$.
\end{enumerate}
\end{enumerate}
}}
\medskip

Note that $\bigp'$ always calls $\bigp$ as a subroutine.  When $S \neq \emptyset, \{j\}$, then $\bigp'(T, S)$ is identical to $\bigp(T, S)$.  On the other hand, when $S$ equals $\emptyset$ or $\{j\}$, then $\bigp'(T, S)$ outputs exactly what $\bigp(T,S)$ would output, but then it calls either $\bigp(T_0, \emptyset)$ or $\bigp(T_1, \emptyset)$.  These two will immediately terminate, so $\bigp'$ has the same output behavior as $\bigp$.  Thus, to show that $p_i(T) = \frac{1}{2} \cdot \left(p_i(T_0) + p_i(T_1)\right)$, it suffices to show that the probability that the path output by $\bigp'(T, \cdot )$ contains index $i$ is $\frac{1}{2} \cdot \left(p_i(T_0) + p_i(T_1)\right)$.

We will show that the probability $\bigp'(T, \cdot)$  makes a call to $\bigp(T_0, \cdot)$ versus $\bigp(T_1, \cdot)$ is exactly~$\frac{1}{2}$.  Next, we will show that the sets it calls $\bigp(T_0, \cdot)$ with are distributed as $\widehat{g}^2$, and similarly for $\bigp(T_1, \cdot)$, so that the recursion works.

Without loss of generality, assume that $x_n$ is the root variable of $T$.  Let $S \subseteq [n-1]$ be any set.  The protocol $\bigp'(T, \bigx)$  can only call $\bigp(T_0, S)$ when $\bigx$ is either $S$ or $S \cup\{n\}$, which happens with probability $\widehat{f}(S)^2 + \widehat{f}(S\cup\{n\})^2$.  By Proposition~\ref{prop:commonfact}, $\widehat{f}(S)^2 + \widehat{f}(S\cup\{n\})^2 = \frac{1}{2}\left(\widehat{g}(S)^2 + \widehat{h}(S)^2\right)$.  In either of these two cases, $\bigp(T_0, S)$ is called with probability proportional to $\widehat{g}(S)^2$, and $\bigp(T_1, S)$ is called with probability proportional to $\widehat{h}(S)^2$.  Thus, the probability that $\bigp(T_0, S)$ is called is
\begin{equation*}
\frac{1}{2}\left(\widehat{g}(S)^2 + \widehat{h}(S)^2\right)
\cdot \frac{\widehat{g}(S)^2}{\widehat{g}(S)^2 + \widehat{h}(S)^2}
= \frac{ \widehat{g}(S)^2}{2}.
\end{equation*}
\noindent
Summing over all sets $S$, the probability that $\bigp(T_0, \cdot)$ is called is exactly $1/2$, and conditioned on this occurring, the probability that $\bigp(T_0, S)$ is called is exactly $\widehat{g}(S)^2$.  A similar argument holds with $T_1$ in place of $T_0$ and $h$ in place of $g$.

Thus, when $\bigp'(T, \bigx)$ calls $\bigp(T_0, \cdot)$, the input to the recursive call is distributed as $\widehat{g}^2$, meaning that the path constructed in the recursive call contains $x_i$ with probability $p_i(T_0)$.  Similarly, when $\bigp'(T, \bigx)$ calls $\bigp(T_1, \cdot)$, the path constructed in the recursive call contains $x_i$ with probability $p_i(T_1)$.  Combining these, $p_i(T) = \frac{1}{2}\left(p_i(T_0) + p_i(T_1)\right).$
\end{proof}

\subsection{A covariance bound for read-$k$ decision trees}\label{sec:readkcov}

In this section, we prove Lemma \ref{thm:covbound}.

\begin{lemma}[Lemma~\ref{thm:covbound} restated.]\label{lem:readkcovlemma}
Let $f:\{-1, 1\}^n \rightarrow \{-1, 1\}$ be computed by a read-$k$ decision tree $T$.  Then $\Cov[T] \leq (k-1)\cdot\Var[f]$.
\end{lemma}
Combining this with Lemma~\ref{thm:covthm} and Lemma~\ref{lem:nozero} yields our FEI bound for read-$k$ decision trees:
\begin{theorem}
If $f:\{-1, 1\}^n \rightarrow \{-1, 1\}$ is computable by a read-$k$ decision tree, then $\ent[\widehat{f}^2] \leq (2 +(2k+2)\cdot\log_2 3)\cdot \Inf[f]$.
\end{theorem}
It is not at all clear whether our upper bound  in Lemma~\ref{thm:covbound} is tight.  Potentially, this bound could be replaced with $\Cov[T] \leq \log_2 k \cdot \Var[f]$.  The tight example of this was presented earlier: let $T$ be the tree given in Figure~\ref{fig:badtree}, only with $l$ layers of dummy variables rather than just two.  Furthermore, suppose that $T'$ is itself read-once.  It is easy to see that $T$ is read-$2^l$ and has tree-covariance $\Cov[T] = l\cdot \Var[f]$.  Thus, in this case, $\Cov[T] = \log 2^l \cdot \Var[f]$.

We will prove Lemma~\ref{thm:covbound} by structural induction on $T$.  As is often the case, we will need to strengthen the inductive hypothesis for the induction to go through.  The reason for this is that the read-$k$ decision tree definition only keeps track of the maximum number of times any variable appears in $T$, whereas we require a more fine-grained accounting of the number of times each variable appears.  For a nonempty subset $S \subseteq [n]$, define $m_T(S)$ to be the maximum over $i \in S$ of the number of times $x_i$ appears in the tree $T$.  For example, if $T$ is read-$k$ then $m_T([n]) \leq k$.  We will prove the following lemma:

\begin{lemma}\label{lem:induct}
Let $T$ be a decision tree which computes $f:\{-1, 1\}^n \rightarrow \{-1, 1\}$.  Then
\begin{equation*}
\Cov[T] \leq \sum_{S \neq \emptyset} (m_{T}(S) - 1) \cdot \hat{f}(S)^2.
\end{equation*}
\end{lemma}

Note that if $T$ is read-$k$, the right-hand side is at most $(k-1)\cdot\sum_{S\neq \emptyset}\widehat{f}(S)^2 = (k-1)\cdot\Var[f]$, the bound we are looking for.

\begin{proof}[Proof of Lemma~\ref{lem:induct}]
We prove this by structural induction on the tree $T$.  The base case we consider is when $T$ queries a single variable.
\paragraph{Base case:}
In this case, the left and right subtrees are constant functions, so their covariance is zero.  For the sum on the right-hand side, any $S$ for which $\widehat{f}(S)$ is nonzero must consist of variables queried by $T$, in which case $(m_T(S) -1 ) \geq 0$.  As a result, the right-hand side is always at least~$0$.

\paragraph{Inductive step:} Suppose the root variable of $T$ is $x_n$.  Let $T_0$ and $T_1$ be the left and right subtrees of $T$, respectively.  For convenience, we will upper-bound $2\cdot \Cov[T]$, which can be written as
\begin{equation*}
2\cdot \Cov[T] = 2\cdot \Cov[g, h] + \Cov[T_0] + \Cov[T_1].
\end{equation*}

We will begin with the first term on the right-hand side.  Let $J$ be the set of coordinates which appear in both $T_0$ and $T_1$.  Because $x_n$ is the root variable, it cannot appear in either $T_0$ or $T_1$, so $J$ is a subset of $[n-1]$.
Then
\begin{align*}
2\cdot \Cov[g,h]
& = \sum_{S\neq\emptyset} 2\cdot\widehat{g}(S)\widehat{h}(S)\\
& = \sum_{\emptyset \neq S \subseteq J} 2\cdot\widehat{g}(S) \widehat{h}(S)\\
& \leq \sum_{\emptyset \neq S \subseteq J}\widehat{g}(S)^2 + \widehat{h}(S)^2,
\end{align*}
where the last line holds because $2ab \leq a^2 + b^2$.

Now we focus on the second term.  Applying the inductive hypothesis to $g$ and $h$ yields
\begin{equation}
\Cov[T_0] + \Cov[T_1] \leq
\sum_{\emptyset \neq S \subseteq  [n-1]}
(m_{T_0}(S)-1) \cdot\widehat{g}(S)^2 + (m_{T_1}(S)-1) \cdot\widehat{h}(S)^2.\label{eq:t0t1expansion}
\end{equation}

For any $S$ in the above sum, we have that $m_{T_0}(S) \leq m_T(S)$.  This is because $T_0$ is a subtree of $T$.  However, when $S \subseteq J$ we get the following improved bound: $m_{T_0}(S)\leq m_T(S)-1$.  This holds because every variable in $S$ is queried at least once in $T_1$, and so it must be queried in $T_0$ at least one fewer time than in the whole of $T$.  Similarly, all of these inequalities hold when $T_0$ is replaced with $T_1$.  Rewriting Equation~\ref{eq:t0t1expansion},
\begin{equation*}
\Cov[T_0] + \Cov[T_1]
\leq \sum_{\emptyset\neq S \subseteq J} (m_{T}(S) - 2)\cdot \left(\widehat{g}(S)^2 + \widehat{h}(S)^2\right)
+ \sum_{\emptyset\neq S \nsubseteq J} (m_{T}(S)-1) \cdot \left(\widehat{g}(S)^2 + \widehat{h}(S)^2\right)
\end{equation*}
Now, if we add $2\cdot\Cov[g, h]$ to this, we see that it will add $1$ to the coefficient of $\widehat{g}(S)^2$ and $\widehat{h}(S)^2$ exactly when $\emptyset \neq S \subseteq J$.  As a result,
\begin{equation*}
2\cdot\Cov[g, h] + \Cov[T_0] + \Cov[T_1]
\leq \sum_{\emptyset \neq S \subseteq [n-1]} (m_{T}(S)-1)\cdot \left(\widehat{g}(S)^2 + \widehat{h}(S)^2\right).
\end{equation*}
The left-hand side is $2\cdot\Cov[T]$.  As for the right-hand side, applying Proposition~\ref{prop:commonfact} shows that it is equal to
\begin{equation*}
2\cdot \sum_{\emptyset \neq S \subseteq [n-1]}
(m_{T}(S)-1)\cdot \left(\widehat{f}(S)^2 + \widehat{f}(S \cup \{n\})^2\right).
\end{equation*}
We would be done, except $\widehat{f}(S\cup\{n\})^2$ should have $m_T(S\cup\{n\})$ as its coefficient, not $m_T(S)$.  However, $m_T(S) \leq m_T(S \cup \{n\})$ always, so we can perform this replacement.  This yields the lemma.
\end{proof}

\section{A composition theorem for protocols}\label{sec:composition}
%We begin by proving a composition theorem for $C$-good protocols.
In this section we prove Theorem \ref{thm:protocolcomp}.  Theorem \ref{thm:protocolcomp} concerns several different functions and their spectral distributions defined with respect to different product distributions.  We assume here familiarity with Fourier analysis for product distributions over the Boolean cube (see \cite{OD13} for an introduction) and briefly review some basic facts and notation used in the proof.

For a Boolean function $f:\{-1,1\}_\mu^n\to \{-1,1\}$, where $\mu=\langle \mu_1,\cdots,\mu_n\rangle$ is a sequence of biases, we think of $\{-1,1\}_\mu^n$ as endowed with the product distribution that sets each bit independently in $\{-1,1\}$ with expectation $\E_\mu[x_i]=\mu_i$ and $\Var_\mu[x_i]=1-\mu_i^2$.  Then the $\mu$-biased Fourier decomposition of $f$ is
$$f = \sum_{S\subseteq [n]} \widetilde{f}(S)\phi_S^{\mu}$$
where $$\phi_S^{\mu}(x) = \prod_{i\in S} \frac{x_i-\mu_i}{\Var_\mu[x_i]},$$
and  $\widetilde{f}(S) = \E_\mu[f \cdot \prod_i{i\in S}].$
Thus, a spectral sample from $\widetilde{h}^2$ is distributed so that each $\bigy$ appears with probability $\widetilde{h}(\bigy)^2.$
%to denote the $\mu$-biased Fourier coefficient of $f$ for a subset $S\subseteq [n]$.
%Given $f:\{-1, 1\}^k \rightarrow \{-1, 1\}$ %be a function over the $\eta$-biased variables $y_1, \ldots, y_k$,
%and functions $g_1, \ldots, g_k :\{-1, 1\}^n \rightarrow \{-1, 1\}$ over variables $x_1, \ldots, x_n$, we show how to construct a protocol for the composed function $h(x_1,\cdots,x_n) = f(g_1(x^1),\cdots,g_k(x^k))$ using protocols
%\ref{thm:} where each variable is relevant for just one of the $g_i$'s,
%Suppose that $\E[g_i(\x)] = \eta_i$ for each $i$, meaning that the bias of the $g$'s matches the bias of $f$'s input coordinates.  Set $h = f(g_1, \ldots, g_k)$, so that $h$'s input bits are the input variables $x_1, \ldots, x_n$.  We let $\bigx \sim \widetilde{h}^2\setminus \emptyset$ and use $\bigx_i$ to denote the restriction of $\bigx$ to the relevant coordinates of $g_i$.  Note that the $\bigx_i$'s are disjoint, and their union is $\bigx$.
%we recover the main theorem of \cite{OT13}

Now we proceed to prove Theorem \ref{thm:protocolcomp}.  Let $P_f$ be a $C$-good protocol for $f$ under $\eta$ and $P_1, \ldots, P_k$ be $C$-good protocols for $g_1, \ldots, g_k$ under $\mu$.  Recall that these protocols are prefix-free.
Now, consider a spectral sample $\bigy \sim \widetilde{h}^2$ and the following protocol $P_h(\bigy)$:

\medskip
\fbox{\parbox{14.5cm}{
\begin{enumerate}
\item Let $S \subseteq [k]$ be the set containing those $i \in [k]$ such that $\bigy_i \neq \emptyset$.
\item Output $P_f(S)$.
\item For each $i \in S$, output $P_i(\bigy_i)$.
\end{enumerate}
}}
\medskip

Because the subprotocols are prefix-free, $P_h(\bigy)$ is a prefix-free encoding of $\bigy$.  %unambiguously communicates the value of $\bigy$.
This is because if one scans the output of $P_h(\bigy)$ from left-to-right, the first prefix which could be output by $P_f(\cdot)$ must actually be the output of $P_f(S)$.  This gives a description of the set $S$, from which one can recover $\bigy_1, \ldots, \bigy_k$ by a similar process.  We will show that if $P_f$ and $P_1, \ldots, P_k$ are efficient, then $P_h$ is efficient as well.  To begin, we will need the following pair of claims:

\begin{claim}\label{claim:gdist}
Conditioned on $\bigy_i \neq \emptyset$, $\bigy_i$ is distributed as $\widetilde{g_i}^2\setminus \emptyset$.
\end{claim}
\begin{claim}\label{claim:fdist}
The set $S$ is distributed as $\widetilde{f}^2 \setminus \emptyset$.
\end{claim}
The proofs of these claims, as well as the basic Fourier analytic facts used to prove them, may be found in the Appendix \ref{sec:biasedprelims}.  We now prove the composition theorem for $C$-good protocols.
\begin{lemma}\label{lem:bigfatlemma}
If $P_i$ is a $C$-good protocol for each $g_i$ and $P_f$ is a $C$-good protocol for  $f$, then $P_h$ is a $C$-good protocol for $h$.
\end{lemma}
\begin{proof}
The expected output size of the protocol is
\begin{align*}
\E[\vert P_h(\bigy)\vert]
&= \E\left[\vert P_f(S) \vert + \sum_{i \in S} \vert P_i(\bigy_i)\vert\right]\\
&= \E\left[\vert P_f(S) \vert + \sum_{i =1}^k \bone[\bigy_i \neq \emptyset] \cdot\vert P_i(\bigy_i)\vert\right].
\end{align*}

First, we upper bound the second term in the expectation.  For a fixed $i$,
\begin{equation}
\E\big[\bone[\bigy_i \neq \emptyset] \cdot |P_i(\bigy_i)|\big]
= \Pr[\bigy_i \neq \emptyset] \cdot \E\left[|P_i(\bigy_i)|\big\vert \bigy_i \neq \emptyset\right].\label{eq:protocolgi}
\end{equation}
From Claim~\ref{claim:gdist}, $\bigy_i$ conditioned on $\bigy_i \neq \emptyset$ is distributed as $\widetilde{g_i}^2\setminus \emptyset$.  Thus, as $P_i$ is a $C$-good protocol for $g_i$, we may upper bound $\E\left[|P_i(\bigy_i)|\big\vert \bigy_i \neq \emptyset\right]$ with the expression in the definition of a $C$-good protocol, except where that definition uses an $\bigy$, we have instead $\bigy_i \big\vert (\bigy_i \neq \emptyset)$.  Note that $\Pr[\bigy_i \neq \emptyset] \cdot \E\left[\vert \bigy_i \vert \big\vert \bigy_i \neq \emptyset\right] = \E[|\bigy_i|]$ and that $\Pr[\bigy_i \neq \emptyset] \cdot \Pr[j \in \bigy_i \big \vert \bigy_i \neq \emptyset] = \Pr[j \in \bigy_i]$.  As a result, the upper bound we get on Equation~\ref{eq:protocolgi} is
\begin{equation*}
C \cdot \left(\E[|\bigy_i|] - \Pr[\bigy_i \neq \emptyset]\right)
+ \sum_j \Pr[j \in \bigy_i] \cdot \log\frac{1}{\Var_\mu[x_j]} + \Pr[\bigy_i \neq \emptyset] \cdot \log \Var_\mu[g_i].
\end{equation*}
Note that $\bigy_i \neq \emptyset$ exactly when $i \in S$.  As a result, summing this over all $i \in [k]$ yields
\begin{align}
\E\left[ \sum_{i \in S} \vert P_i(\bigy_i)|\right]\leq~
&C \cdot \left(\E[|\bigy|] - \E[|S|]\right)\nonumber\\
&+ \sum_{j\in[n]} \Pr[j \in \bigy] \cdot \log\frac{1}{\Var_\mu[x_j]} + \sum_{i\in[k]} \Pr[i \in S] \cdot \log \Var_\mu[g_i].\label{eq:rightupper}
\end{align}

For the first term in the expectation, we know by Claim~\ref{claim:fdist} that the random variable $S$ defined in the protocol is distributed according to~$\widetilde{f}^2\setminus \emptyset$.  Thus, because $P_f$ is a $C$-good protocol,
\begin{equation*}
\E[P_f(S)] \leq C \cdot (\E[\vert S \vert] - 1) + \sum_{i \in [k]} \Pr[i \in S] \cdot \log\frac{1}{\Var_\eta[y_i]} + \log \Var_\eta[f].
\end{equation*}
Note that $\Var_\eta[y_i] = \Var_\mu[g_i]$ and $\Var_\eta[f] = \Var_\mu[h]$.  As a result, adding these together yields
\begin{equation*}
\E[|P(\bigy)|] \leq C \cdot (\E[|\bigy|] - 1) + \sum_{j \in [n]}\Pr[j \in \bigy] \cdot \log \frac{1}{\Var_\mu[x_j]} + \log \Var_\mu[h],
\end{equation*}
which yields the theorem.
\end{proof}

\bibliographystyle{alpha}
\bibliography{wright}

\appendix
\section{Decision tree proofs}\label{sec:readk-prelim}

We will repeatedly use the following proposition, which relates the Fourier coefficients of $f$ to the Fourier coefficients of its subfunctions $g$ and $h$.

\begin{proposition}
Let $f$ be computed by a decision tree $T$ whose left and right subfunctions are $g$ and $h$, respectively.  If $x_i$ is at the root of $T$ and $S$ is any subset of $[n] \setminus \{i\}$, then
\begin{equation*}
\widehat{f}(S)^2 + \widehat{f}(S \cup \{i\})^2 = \frac{1}{2}\left(\widehat{g}(S)^2 + \widehat{h}(S)^2\right).
\end{equation*}
\end{proposition}
\begin{proof}
Write $f$ as
\begin{equation*}
f = \left(\frac{1 + x_i}{2}\right) g + \left(\frac{1- x_i}{2}\right) h.
\end{equation*}
For any $S \subseteq [n]\setminus \{i\}$, $\widehat{f}(S) = \frac{1}{2}(\widehat{g}(S) + \widehat{h}(S))$ and $\widehat{f}(S \cup \{i\}) = \frac{1}{2}(\widehat{g}(S) - \widehat{h}(S))$.  As a result,
\begin{equation*}
\widehat{f}(S)^2 + \widehat{f}(S \cup \{i\})^2 = \frac{1}{2}\left(\widehat{g}(S)^2 + \widehat{h}(S)^2\right).\qedhere
\end{equation*}
\end{proof}

We will also use the following proposition, which relates the influences of $f$ to the influences of its subfunctions.

\begin{proposition}
Assume the setup from Proposition~\ref{prop:commonfact}.  Then for a coordinate $j \neq i$,
\begin{equation*}
\Inf_j[f] = \frac{1}{2}\cdot(\Inf_j[g] + \Inf_j[h]).
\end{equation*}
\end{proposition}
\begin{proof}
\begin{align*}
\Inf_j[f]
&= \Pr[f(\x) \neq f(\x^{\oplus j})] \\
&= \frac{1}{2} \Pr[g(\x) \neq g(\x^{\oplus j})] + \frac{1}{2} \Pr[h(\x) \neq h(\x^{\oplus j})]
= \frac{1}{2}\left( \Inf_j[g] + \Inf_j[h]\right).\qedhere
\end{align*}
\end{proof}

\section{Proof of Lemma~\ref{lem:nozero}}\label{sec:nozero}

In this section, we give a proof of Lemma~\ref{lem:nozero}, which was implicit in~\cite{OWZ11}; the proof we give here, included for completeness, is essentially the same.
%basically the same as theirs.
First, we have the following lemma:
\begin{lemma}\label{lem:edgeiso}
Let $f:\{-1, 1\}^n \rightarrow \{-1, 1\}$ and write $\hat{f}(\emptyset)^2 = 1-\eps$.  Then $2 \cdot \Inf[f] \geq h(\eps)$, where $h(\cdot)$ is the binary entropy function.
\end{lemma}
\begin{proof}
First, we may assume $\eps \neq 0, 1$, otherwise the result is trivial.  Now, $(1-\eps)\log\frac{1}{1-\eps} \leq \frac{1}{\ln 2}\eps \leq 2\eps$, so
\begin{equation*}
h(\eps) = \eps \log \frac{1}{\eps} + (1 - \eps)\log\frac{1}{1-\eps} \leq \eps \log \frac{1}{\eps} + 2\eps.
\end{equation*}
By Proposition 2 of~\cite{OWZ11}, the right-hand side is at most $2\cdot \Inf[f]$, and the lemma follows.
\end{proof}

\begin{lemma}[Restatement of Lemma~\ref{lem:nozero}]
Suppose there is an almost prefix-free protocol for $\bigx$ with length $B$ and alphabet $\Sigma$.  Then $\ent[\bigx] \leq \log_2 |\Sigma| \cdot B + 2\cdot \Inf[f]$.
\end{lemma}
\begin{proof}
Write $\hat{f}(\emptyset)^2 = 1- \eps$. If $\eps = 0$ then $\ent[\bigx] = 0$, so the lemma follows.  Otherwise, let $\bigy$ be the indicator that $\bigx = \emptyset$.  Then $\ent[\bigx \vert \bigy = 0] \leq  \log_2 |\Sigma|~B/\eps$ by the source coding theorem, as the protocol outputs $B/\eps$ characters on average conditioned on $\bigx$ being nonempty.
\begin{align*}
\ent[\bigx]
& = \ent[\bigx, \bigy]\\
& = \ent[\bigy] + \ent[\bigx \vert \bigy]\tag{conditional entropy}\\
& = \ent[\bigy] + (1-\eps) \cdot \ent[\bigx\vert \bigy = 1] + \eps \cdot \ent[\bigx\vert \bigy = 0]\\
& \leq \ent[\bigy] +\log_2|\Sigma|\cdot B \tag{using $\ent[\bigx \vert \bigy = 1] = 0$}
\end{align*}
Because $\bigy$ is a $(1-\eps)$-biased random bit, $\ent[\bigy] = h(\eps)$, where $h(\cdot)$ is the binary entropy function.  Thus, we may apply Lemma~\ref{lem:edgeiso} and get that $\ent[\bigx] \leq 2\cdot \Inf[f] + \log_2|\Sigma| \cdot B$.
\end{proof}

\section{Parallelizing the Protocol}\label{sec:parallel}

%Lemma~\ref{lem:bigfatlemma} almost gets us Ryan and Li-Yang's composition theorem, but Lemma~\ref{lem:bigfatlemma} requires a $C$-good protocol for $f$ and the $g_i$'s, whereas their composition theorem only requires that these functions have low spectral entropy, a weaker condition.  This problem can be circumvented if instead we consider protocols which communicate many independent samples from the spectral distribution.  For example,
The performance of Shannon's code gives the following guarantee:
\begin{fact}\label{fact:shannon}
Let $\bigx^1, \ldots, \bigx^{t}$ be $t$ independent samples drawn from $\widetilde{f}^2\setminus \emptyset$ .  Then there is a prefix-free protocol $P_f^t$ for which
\begin{equation*}
\E\big[|P_f^t(\bigx^1, \ldots, \bigx^{t})|\big] \leq t\cdot \textbf{H}\left[\widetilde{f}^2\setminus \emptyset\right] +1.
\end{equation*}
\end{fact}
In other words, the average number of bits used per copy of $\bigx$ is $1/t$ more than the theoretical best.  In the limit as $t$ tends to $\infty$, the excess number of bits tends to~$0$.  Using this, we will show that the protocol from Section~\ref{sec:composition} may be analyzed as if the subprotocols are optimally efficient.  We will do this by showing an efficient protocol to communicate sets $\bigy^1, \ldots, \bigy^t \sim \widetilde{h}^2\setminus \emptyset$ which are chosen independently.  As before, use $\bigy^j_i$ to denote the restriction of $\bigy^j$ to the coordinates relevant to $g_i$.  We will assume that we have the efficient protocol $P_f^t$ guaranteed by Fact~\ref{fact:shannon}.  In addition, for each $i \in [k]$ and $m \in [t]$ we will use the protocol $P_i^m$ which Fact~\ref{fact:shannon} guarantees will efficiently communicate $m$ samples from $\widetilde{g_i}^2\setminus \emptyset$.  Now, consider the following protocol $P^t_h(\bigy^1, \ldots, \bigy^t)$:

\medskip
\fbox{\parbox{14.5cm}{
\begin{enumerate}
\item For each $j \in [t]$, let $S^j \subseteq [k]$ be the set containing those $i$ such that $\bigy_i^j \neq \emptyset$.
\item Output $P_f^t(S^1, \ldots, S^t)$.
\item For each $i \in [k]$:
\begin{enumerate}
\item Let $j_1, \ldots, j_m$ be the indices of the nonempty $\bigy_i^j$ (in order).
\item If $m \neq 0$, output $P_i^m(\bigy_i^{j_1}, \ldots, \bigy_i^{j_m})$.  Otherwise, output nothing.
\end{enumerate}
\end{enumerate}
}}
\medskip

The following lemma, which may be compared to Proposition 3.2 in \cite{OT13},
gives the performance of this protocol and suffices to recover their composition theorem for
entropy.
%I think they're essentially the same thing, except using this lemma you can only upper bound the entropy of $\widetilde{h_i}^2\setminus\emptyset$, whereas they get equality.
\begin{lemma}
Let $S$ be distributed as in the protocol from Section~\ref{sec:composition}.  In the limit as $t\rightarrow \infty$,
\begin{equation*}
\frac{1}{t} \cdot\E\left[|P_h^t(\bigy^1, \ldots, \bigy^t)|\right]
= \textbf{H}\left[\widetilde{f}^2\setminus \emptyset\right] + \sum_{i \in [k]} \Pr[i \in S]\cdot \textbf{H}\left[\widetilde{g_i}^2\setminus\emptyset\right].
\end{equation*}
\end{lemma}

\begin{proof}[Proof sketch]
Fix a coordinate $i \in [k]$ and consider the number $m$ of nonempty $\bigy_i^j$s.  From Claim~\ref{claim:gdist}, we know that if $\bigy_i^j$ is nonempty, then it is distributed as $\widetilde{g_i}^2\setminus \emptyset$.  As a result, for a fixed value of $m$, Fact~\ref{fact:shannon} tells us that the expected number of bits that $P_i^m$ outputs per $\bigy_i^j$ is at most $1/m$ in excess of $\textbf{H}\left[\widetilde{g_i}^2\setminus\emptyset\right]$.  Now, $m$ is distributed as $\mathrm{Bin}(t, r)$, where $r$ is a probability independent of $t$.  Thus, by taking $t \rightarrow \infty$ the expectation of $1/m$ (when $m$ is nonzero) will tend towards~$0$.  As a result, we may assume that $\textbf{H}\left[\widetilde{g_i}^2\setminus\emptyset\right]$ bits are used in expectation to communicate each nonzero $\bigy_i^j$.  A similar argument shows that we may assume that $\textbf{H}[\widetilde{f}^2\setminus \emptyset]$ bits are used in expectation to communicate each $S^j$.

Aside from packaging the different sets together when calling the subprotocols, the protocol acts as $t$ independent copies of the protocol from Section~\ref{sec:composition}.  Let us focus on the case when $j = 1$.  Then the expected number of bits spent outputting the sets for which $j=1$ is
\begin{equation*}
\textbf{H}\left[\widetilde{f}^2\setminus \emptyset\right] + \sum_{i \in [k]} \Pr[i \in S^1]\cdot \textbf{H}\left[\widetilde{g_i}^2\setminus\emptyset\right].
\end{equation*}
As $S^1$ is distributed identically to $S$ in the protocol from Section~\ref{sec:composition}, we may replace the event $i \in S^1$ with $i \in S$.  Averaging this over all $j \in [t]$ yields the lemma.
\end{proof}

\section{Proofs of Claims}\label{sec:biasedprelims}
First, we recall several basic facts regarding $\mu$-biased Fourier analysis.
For $S\neq T$ and $S\neq \emptyset$, we have $\E_\mu[\phi_S^{\mu}] =0$ and
$\E_\mu[\phi_S^\mu \cdot \phi_T^\mu] =0$.  We also have Parseval's inequality, which states that
for $f:\{-1,1\}^n \to \R$, the equality $\sum_{S\subseteq [n]}\widetilde{f}(S)^2 = \E_\mu[f^2]$ holds.

We now prove the following proposition, from which the claims follow immediately.
\begin{proposition}\label{prop:fourier}
Given the setup of the first protocol,
\begin{equation*}
\widetilde{h}(\bigy) = \widetilde{f}(S)\prod_{i \in S}\frac{\widetilde{g}(\bigy_i)}{\sigma_i}.
\end{equation*}
\end{proposition}
\begin{proof}
Let $\eta_i=\bfe_\mu[g_i]$ and $\sigma^2_i=\Var_\mu[g_i]$.  Let $S$ be as defined in the protocol.  Then
\begin{align*}
\widetilde{h}(\bigy)
& = \E_{\bx\sim \mu}[h(\bx)\cdot\phi_{\bigy}^\mu(\bx)]\\
& = \E_\bx\left[f(g_1(\bx), \ldots, g_k(\bx)) \cdot \prod_{j \in S} \phi_{\bigy_j}^\mu(\bx)\right]\\
& = \sum_{T \subseteq [k]}\widetilde{f}(T)
			\E_\bx\left[\phi_{T}^\eta(g_1(\bx), \ldots, g_k(\bx)) \cdot \prod_{j \in S} \phi_{\bigy_j}^\mu(\bx)\right]\\
& = \sum_{T \subseteq [k]}\widetilde{f}(T)
			 \E_\bx\left[\prod_{i\in T}\left(\frac{g_i(\x) -\eta_i}{\sigma_i}\right) \prod_{j \in S} \phi_{\bigy_j}^\mu(\bx)\right].
\end{align*}
A standard calculation shows that the expectation is nonzero only if $S = T$.  In this case, the expectation is equal to\
\begin{equation*}
\prod_{i \in S}\E_\bx\left[\left(\frac{g_i(\bx)-\eta_i}{\sigma_i}\right) \cdot \phi_{\bigy_i}^\mu(\bx)\right]
= \prod_{i \in S}\frac{\widetilde{g}(\bigy_i)}{\sigma_i},
\end{equation*}
where the equality holds because $\bigy_i$ is nonempty, so the shift by $\eta_i$ doesn't affect the calculation.  The proposition now follows.
\end{proof}

Now we prove the claims:
\begin{proof}[Proof of Claim~\ref{claim:gdist}]
Condition $\bigy$ on $\bigy_i \neq \emptyset$ and on any values for $\bigy_1, \ldots, \bigy_{i-1}, \bigy_{i+1}, \ldots, \bigy_{k}$.  Then by Proposition~\ref{prop:fourier}, $\bigy_i$ is distributed as $\widetilde{g_i}^2\setminus \emptyset$.  As this holds conditioned on any values for the $\bigy_j$'s, $j \neq i$, this also holds conditioned only on $\bigy_i \neq \emptyset$.
\end{proof}

\begin{proof}[Proof of Claim~\ref{claim:fdist}]
First, because $f$ and $h$ have the same mean, they also have the same variance, i.e.
\begin{equation*}
\sum_{\bigy\neq\emptyset}\widetilde{h}^2(\bigy) = \sum_{S\neq\emptyset}\widetilde{f}^2(S).
\end{equation*}
Next, fix a particular value of $S \subseteq [k]$, $S \neq \emptyset$.  The sets $\bigy$ for which the protocol selects this particular $S$ are those for which $\bigy_i \neq \emptyset \iff i \in S$.  Then the probability $S$ is selected is just the sum over these sets:
\begin{align*}
\sum_{\bigy:\bigy_i \neq \emptyset \iff i \in S}\widetilde{h}(\bigy)^2
& = \sum_{\bigy:\bigy_i \neq \emptyset \iff i \in S}
		\widetilde{f}(S)^2\prod_{i \in S}\frac{\widetilde{g}(\bigy_i)^2}{\sigma_i^2}\tag{by Proposition~\ref{prop:fourier}}\\
& = \widetilde{f}(S)^2\prod_{i \in S}\sum_{\bigy_i \neq \emptyset}\frac{\widetilde{g}(\bigy_i)^2}{\sigma_i^2}\\
& = \widetilde{f}(S)^2.
\end{align*}
Combining these two facts yields the claim.
\end{proof}

\section{Small influence counterexample}\label{app:biginf}

Suppose we could prove Theorem~\ref{thm:depthd} without the restriction on the function's total influence, i.e. the following statement:
\begin{conjecture}\label{conj:lowinf}
Suppose $f:\{-1, 1\}^n \rightarrow \{-1, 1\}$ is computable by a decision tree with expected depth $d$, and let $\bigx \sim \widehat{f}^2$.  Then $\ent[\bigx] \leq C\cdot d \cdot \Inf[f]$, for some absolute constant $C$.
\end{conjecture}
This appears to be a weaker conjecture than the FEI conjecture.  However, in this section we will show that this statement implies the FEI conjecture, at least for functions with sufficiently large influence.
\begin{proposition}\label{prop:dtreduction}
Suppose Conjecture~\ref{conj:lowinf} were true for some constant $C$.  Let $f:\{-1, 1\}^n \rightarrow \{-1, 1\}$, and let $\bigx \sim \widehat{f}^2$.  If $\Inf[f] \geq \log(n)$, then $\ent[\bigx] \leq C' \cdot \Inf[f]$, where $C'$ is some other absolute constant.
\end{proposition}
Although this only shows that Conjecture~\ref{conj:lowinf} implies a restricted form of the FEI conjecture, this restricted form does not appear to be especially easier than the full FEI conjecture.  Thus, the restriction in Theorem~\ref{thm:depthd} that $f$ have large influence is a natural one.

Let $f:\{-1, 1\}^n \rightarrow \{-1, 1\}$ have $\Inf[f] \geq \log n$.  We prove Proposition~\ref{prop:dtreduction} by ``hiding'' $f$ in a low expected-depth decision tree.  The resulting decision tree still has low expected-depth, and its spectral entropy and total influence terms are roughly proportional to $f$'s.  Thus,  applying Conjecture~\ref{conj:lowinf} to the decision tree shows that $f$ itself satisfies the FEI conjecture.

\begin{proof}

Let $f:\{-1, 1\}^n \rightarrow \{-1, 1\}$ have $\Inf[f] \geq \log n$, and let $\bigx \sim \widehat{f}^2$.  We begin with the assumption that $f$ is balanced, i.e. that $\E[f] = 0$, and we will later reduce the general case to this case.
For simplicity, assume that $n$ is a power of two.  Consider the new function $g(x, y)$ defined as
\begin{equation*}
g(x, y) = \left\{
	\begin{array}{cl}
		f(x) & \text{if } AND(y_1, \ldots, y_k) = -1,\\
		1 & otherwise.
	\end{array} \right.
\end{equation*}

Pictorially, refer to Figure~\ref{fig:depthd}, where $f(x)$ is computed by some decision tree.  Since $f$ can be trivially computed by a decision tree of depth $n$, the decision tree pictured computes $f$ with expected depth at most $2 + n/2^k$.  By choosing $k = \log_2 n$, this decision tree has expected depth~$3$.
\begin{figure}
\centering
\includegraphics{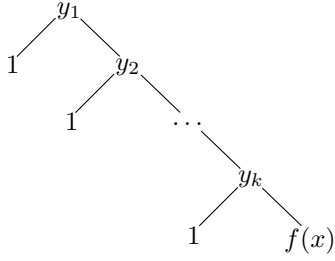}
\caption{A decision tree computing $g(x, y)$.}\label{fig:depthd}
\end{figure}

Each variable $y_i$ is influential only when the rest of the $y_j$'s are~$-1$ and $f(x) = -1$ (which happens half of the time because $f$ is balanced), so the influence of each $y_i$ is exactly $1/2^k$.  Each of the $x_i$ variables is influential only when all of the $y_i's$ are~$-1$, so the influence of variable $x_i$ on $g$ is exactly $\Inf_i[f]/2^k$.  As a result,
\begin{equation*}
\Inf[g] = \frac{k + \Inf[f]}{2^k} = \frac{\log n + \Inf[f]}{n}.
\end{equation*}

To compute the entropy, we can first write $g$ as
\begin{equation*}
g(x, y) = 1 - \frac{1}{2^k} \sum_S \chi_S(y) \left(\frac{1}{2} + \frac{f(x)}{2}\right).
\end{equation*}
From this, we can easily read off some of the Fourier coefficients of $g$: if $S \subseteq [k]$ and $T \subseteq [n]$ are both nonempty, then $\widehat{g}(S, T) = -\widehat{f}(T)/2^{k+1}$.  As a result, if $\bigx' \sim \widehat{g}^2$, then we can lower bound $\ent[\bigx']$ by summing over the terms in the entropy formula corresponding to these subsets:
\begin{align*}
\ent[\bigx']
& \geq \sum_{S, T \neq \emptyset}
	\frac{\widehat{f}(T)^2}{2^{2k+2}}\log\left(\frac{2^{2k+2}}{\widehat{f}(T)^2}\right)\\
& \geq \sum_{T} \frac{\widehat{f}(T)^2}{2^{k+3}}
	\log\left(\frac{2^{2k+2}}{\widehat{f}(T)^2}\right)\\
& \geq \sum_{T} \frac{(2k+2) \cdot \widehat{f}(T)^2}{2^{k+3}}
	+ \sum_{T} \frac{\widehat{f}(T)^2}{2^{k+3}} \log\left(\frac{1}{\widehat{f}(T)^2}\right)\\
& = \frac{2k+2 + \ent[\bigx]}{2^{k+3}} = \frac{2\log n + 2 + \ent[\bigx]}{8n}.
\end{align*}
Here the second inequality follows because the sum is over $2^k-1 \geq 2^{k-1}$ sets $S$ and because $\widehat{f}(\emptyset) = 0$.  The second-to-last equality follows because $f$ is mean-zero, so $\sum_T \widehat{f}(T)^2 = 1$.

Now, applying Conjecture~\ref{conj:lowinf} to $g$, we have that
\begin{equation*}
\frac{3\log n + \ent[\bigx]}{8n} \leq \ent[\bigx'] \leq C\cdot 3 \cdot\Inf[g] = C \cdot 3 \cdot \frac{\log n + \Inf[f]}{n}.
\end{equation*}
This can be rearranged as
\begin{equation*}
\ent[\bigx] \leq (24C - 3) \cdot \log n + 24C\cdot \Inf[f].
\end{equation*}
Thus, if $\Inf[f] \geq \log n$, then $\ent[\bigx] \leq C' \Inf[f]$, where $C' = 48C - 3$.

Now, if $f$ is not balanced, consider the function $g(x_1, \ldots, x_n, x_{n+1}) = x_{n+1} \cdot f(x)$.  Then $g$ is balanced, has the same Fourier entropy as $f$, and $\Inf[g] = \Inf[f] + 1$.  As we have just shown,
\begin{align*}
\ent[\bigx] &= \ent[\widehat{g}^2]\\
&\leq C'\cdot \Inf[g]\\
&= C'\cdot (\Inf[f] +1)\\
&\leq C' \cdot (\Inf[f] + \log n)\\
&\leq (C' + 1)\cdot \Inf[f].
\end{align*}
Here, the last inequality uses the fact that $\log n \leq \Inf[f]$.  Thus, $f$ satisfies the FEI conjecture with constant $C'+1$.
\end{proof}

%\section{FEI+ as protocol}
%This section is kind of just for my benefit\ldots. Maybe we'll end up needing some of it in the preliminaries.
%The FEI$^+$ conjecture is stated in \cite{OT13} as
%\begin{conjecture}
%$$ \sum_{S\neq \emptyset} \widetilde{f}(S)^2\log \frac{\prod_{i\in S} \sigma^2_i}{\widetilde{f}(S)^2}
%\leq C \cdot \sum_{S\neq \emptyset} \widetilde{f}(S)^2 (|S|-1).$$
%\end{conjecture}
%Equivalently, we have
%%\begin{eqnarray*}
%\sum_{S\neq \emptyset} \widetilde{f}(S)^2 \log \frac{1}{\widetilde{f}(S)^2}
%& \leq & C\cdot \sum_{S\neq \emptyset} \widetilde{f}(S)^2(|S|-1) + \sum_{S\neq \emptyset} \widetilde{f}(S)^2 (\sum_{i\in S} \log \frac{1}{\sigma^2_i}) \\
%& = & C\cdot \sum_{S \neq \emptyset} \widetilde{f}(S)^2(|S|-1) + \sum_{i=1}^n \sum_{S\ni i} \widetilde{f}(S)^2 \log \frac{1}{\sigma^2_i}.
%\end{eqnarray*}
%For $\bigy \sim \widehat{f}(S)^2 \setminus \emptyset$ we have:
%$$\sum_{S\neq \emptyset} \widetilde{f}(S)^2 \log \frac{1}{\widetilde{f}(S)^2} = H(\bigy) - \log \Var_\mu(f).$$
%Thus, the conjecture may be equivalently stated as:
%$$ H(\bigy) \leq C\cdot (\E[ |Y|] - 1) + \sum_i \Pr[i\in \bigy] \log \frac{1}{\sigma^2_i} + \log \Var_\mu (f).$$
\end{document}